\documentclass[12pt, a4paper]{article}
\usepackage{amsfonts, amssymb}
\usepackage{graphicx}
\usepackage{latexsym,amsmath,color,array}

\usepackage{natbib}
\usepackage{enumerate}
\usepackage{bbding}
\usepackage{amssymb}
\usepackage{amsmath}
\usepackage{graphicx}
\usepackage{amsmath}
\usepackage{bbm}
\usepackage{mathtools}
 \usepackage{color}
 \usepackage{array}
 \usepackage{subfigure}
 \usepackage{multirow}
 \usepackage[dvipsnames]{xcolor}
 \usepackage{makecell}
 \usepackage{colortbl}

\def\1{\mathbb{I}}

\newcommand{\Pn}{\mathbb{P}_n}
\newcommand{\Gn}{\mathbb{G}_n}
\newcommand{\R}{\mathbb{R}}

\newcounter{thm}[section]
\newcounter{appen}[section]
\newcounter{assum}[section]
\newcounter{rmk}[section]
\newcounter{lem}[section]

\newtheorem{theorem}[thm]{Theorem}
\newtheorem{lemma}[lem]{Lemma}
\newtheorem{remark}[rmk]{Remark}
\newtheorem{assumption}[assum]{Assumption}

\newenvironment{proof}[1][Proof]{\noindent \textbf{#1.}
}{\rule{0.5em}{0.5em}}

\begin{document}

\title{Semiparametric multi-parameter regression survival modelling}
\author{Kevin Burke\footnote{Department of Mathematics and Statistics, University of Limerick; kevin.burke@ul.ie} \hspace{3cm}
Frank Eriksson\footnote{Section of Biostatistics, University of Copenhagen, Denmark.} \hspace{3cm}
C. B. Pipper\footnote{Section of Biostatistics, University of Copenhagen, Denmark.} }
\date{}

\maketitle

\begin{abstract}
We consider a log-linear model for survival data, where both the location and scale parameters depend on covariates and the baseline hazard function is completely unspecified. This model provides the flexibility needed to capture many interesting features of survival data at a relatively low cost in model complexity. Estimation procedures are developed and asymptotic properties of the resulting estimators are derived using empirical process theory. Finally, a resampling procedure is developed to estimate the limiting variances of the estimators. The finite sample properties of the estimators are investigated by way of a simulation study, and a practical application to lung cancer data is illustrated.

\smallskip

{\bf Keywords.} Counting processes; Empirical processes; Log-linear failure time model; Multi-parameter regression; Semiparametric regression; Survival data.

\end{abstract}

\qquad

\newpage

\section{Introduction}
In the context of survival analysis, we often consider log-linear models of the form \\$\log T = \mu + \sigma e$, where $\mu$ and $\sigma$ are location and scale parameters, respectively, and $e$ is a random error with an assumed parametric distribution on the real numbers. The familiar accelerated failure time model then arises by setting $\mu = - \beta^T X$ where $\beta$ and $X$ are, respectively, vectors of regression coefficients and covariates (\citet[chap.~3]{kalbprent:2002} and \citet[chap.~6]{lawless:2003}. As discussed in \citet{burmac:2017}, taking a multi-parameter regression approach (i.e., allowing both $\mu$ and $\sigma$ to depend on covariates simultaneously) offers an intuitive and simple way of modelling complicated phenomena. For instance, the phenomenon of crossing survival curves is directly linked to the concentration of events at a given location which is governed by the scale parameter $\sigma$.

A limitation of fully parametric approaches is that the assumed baseline hazard may not always be realistic in practice. Thus, we propose to further extend the log-linear multi-parameter regression model by allowing the baseline hazard to vary freely. This semiparametric model therefore brings together the flexibility of multi-parameter regression with additional robustness afforded by relaxing the assumption of a parametric error distribution. The proposed extension not only generalises multi-parameter regression to semiparametric status but also generalises the semiparametric accelerated failure time model to multi-parameter regression status; the latter fact is noteworthy given that the semiparametric accelerated failure time model has been considered by many authors over the years \citep{miller:1976, prentice:1978, buckjam:1979, tsiatis:1990, ritov:1990, laiandying:1991, ying:1993, linetal:1998, jinetal:2003}; see \citet[chap.~8]{martscheike:2007} for a summary of developments in this area.

Other examples of semiparametric multi-parameter regression models, which differ from the model developed in this
paper, exist in the literature. \citet{chenjewell:2001} considered a model which combined the
semiparametric accelerated failure time model and
\citeauthor{cox:1972}'s (\citeyear{cox:1972}) proportional hazards
model and, therefore, had two regression components; the model also
contained the lesser-known accelerated hazards model
\citep{chenwang:2000} as a special case. \citet{scheikezhang:2002,
  scheikezhang:2003} developed a different hybrid model which
incorporated the Cox model and the Aalen model \citep{aalen:1980}
leading to two regression components. Somewhat closer to our work is
that of \citet{zenglin:2007} who considered transformation models with
a covariate-dependent scale parameter and, therefore, like us, had
regression components corresponding to location and scale. However,
whereas their transformation was unspecified with a parametric
baseline distribution, we, conversely, focus on the log-transformation
with an unspecified baseline distribution.

From a practical perspective, inference based on the semiparametric accelerated failure time model has historically been somewhat cumbersome. This is partly due to the non-smooth nature of the estimating equations involved, but, more importantly, the precision of the resulting estimators does not lend itself to direct (i.e., plug-in) estimation due to its intractability. However, with recent resampling techniques, this is no longer an obstacle, and, specifically, we adapt the method of \citet{zeng08} to our setting to facilitate inference for the regression coefficients. Moreover, we expand \citeauthor{zeng08}'s approach to obtain the variance of the cumulative hazard estimator, and combine this with modern empirical process theory which permits straightforward inference for any functional of interest without the need for resolving estimating equations.

\section{Model}

\subsection{Specification and interpretation}

In line with the classical formulation of the accelerated failure time
model (cf.~\citet[chap.~3]{kalbprent:2002}) we specify a
regression model for $\log T$ with $T$ denoting the failure time. In
particular, for the $i$th individual, $i=1,\ldots,n$, we assume that
$$
\log T_{i}=\mu_{i}+\sigma_{i}e_{i},
$$
where the location and scale parameters, $\mu_{i}$ and $\sigma_{i}$, are related to $p+q$ covariates,  $(X_{i}^T,Z_{i}^T)^T$, via
\begin{align*}
&\mu_{i}=-\beta^{T}X_{i},\\
&\sigma_{i}=\exp(-\gamma^{T}Z_{i}),
\end{align*}
where $\beta$ and $\gamma$ are vectors of regression coefficients. The error terms, $e_1,\ldots,e_n$, are assumed to be independent and identically distributed with
cumulative hazard function $A(\cdot)$ which will be unspecified in our work.


The conditional quantile function for this model is given by
\begin{align*}
Q_i(\pi) = \exp(\mu_i) Q_0(\pi)^{\sigma_i},
\end{align*}
where $\pi \in[0,1]$ and $\log\{Q_0(\pi)\}= A^{-1}\{-\log(1-\pi)\}$ is the quantile function for the error distribution, i.e., $Q_0(\pi)$ is the quantile function for a baseline individual. Consider individuals $i$ and $j$ whose respective $X$ and $Z$ vectors are denoted by $X_{i},X_{j},Z_{i},Z_{j}$.  The ratio of their quantile functions is then
\begin{align*}
\frac{Q_j(\pi)}{Q_i(\pi)} = \exp\{-\beta^{T}(X_{j}-X_{i})\}Q_0(\pi)^{\exp(-\gamma^T Z_{j})\,-\,\exp(-\gamma^T Z_{i})}.
\end{align*}
This quantile ratio provides insight into the interpretation of the location and scale regression coefficients and, indeed, can be used in practical applications to quantify the overall effect of a given covariate on lifetime. We immediately see that when $\gamma^T (Z_{j}-Z_{i})=0$, the quantile ratio reduces to the usual accelerated failure time constant, $\exp\{-\beta^{T}(X_{j}-X_{i})\}$, so that the effect of covariates is quantile-independent, i.e., it applies across the whole lifetime, and, for example, $\beta^{T}(X_{j}-X_{i})>0$ implies reduced lifetime. Hence, the proposed model directly extends the accelerated failure time model, providing a lack of fit test of accelerated failure time effects.

It is worth noting that, since $Q_0(\pi)$ is an increasing function on $[0, \infty[$~, the quantile ratio decreases with $\pi$ for $\gamma^T (Z_{j}-Z_{i}) >0$, increases  with $\pi$  for $\gamma^T (Z_{j}-Z_{i}) < 0$, and, in both cases, equals one for some $\pi$ value if $\lim_{\pi\rightarrow1}Q_0(\pi) = \infty$.  Therefore, when $\gamma^T (Z_{j}-Z_{i}) \ne 0$, the model implies crossing quantile functions and, hence, crossing survivor functions.

\subsection{Derivation of estimation equations\label{sec:esteq}}

We now reformulate the model in a counting process framework \citep{andersen93} to adopt potential right censoring, and to derive estimation equations based on the resulting intensity processes. For this we denote by $C_{i}$ the censoring time, $\tilde{T}_{i}=C_i\wedge T_i$ the observed event time, and $\Delta_i=I(T_i\leq C_i)$ the failure indicator.
With these quantities in place, the counting and at risk processes are, respectively, defined as
\begin{align*}
N_i(t)&=\Delta_iI(\log \tilde{T}_i\leq t),\\
Y_i(t)&=I(\log \tilde{T}_i\geq t).
\end{align*}

We note that the hazard rate of $\log T_{i} $ is given by
\begin{align*}
\alpha_i(t)&=\alpha\left\{\sigma_{i}^{-1}(t-\mu_{i})\right\}\sigma_{i}^{-1},
\end{align*}
where $\alpha$ denotes the derivative of $A$. Consequently, with independent right censoring \citep{andersen93}, the intensity process of $N_{i}(t)$ is given by
$Y_{i}(t)\alpha_{i}(t)$. Furthermore, for some monotone increasing function, $g_i$, the time-transformed counting process
\begin{align*}
N_i^*(t)&=N_i\{g_i(t)\}
\end{align*}
has intensity
\begin{align*}
\lambda_i^*(t)&=Y_i^*(t)g_i'(t)\alpha_i\{g_i(t)\}
\end{align*}
where $Y_i^*(t)=Y_i\{g_i(t)\}$. In particular, with
$g_i(t)=\sigma_{i}t+\mu_{i}$,  we have that $N_i^*(t)$
has intensity $Y_i^*(t)\alpha(t)$. This observation motivates a Nelson-Aalen type estimator of $A$, that is,  for a given value of $\theta=(\beta^T,\gamma^T)^T$, we estimate $A$ by
\begin{align*}
\hat{A}_{n}(t,\theta)&=\sum_{i=1}^n\int_{-\infty}^t\frac{dN_{i}^*(s)}{\sum_{j=1}^n Y_j^*(s)}.
\end{align*}

To estimate the regression parameters, $\theta$, we propose the use of a likelihood based approach. For given $A$, the score function for $\theta$ is given by
$$
 \sum_{i=1}^n\int_{-\infty}^{\infty}D_{\theta}\log\{\alpha_i(s)\}\{dN_i(s)-Y_i(s)\alpha_i(s)ds\}.\\
$$
where $D_{\theta} = (\partial/\partial\beta_1, \ldots, \partial/\partial\beta_p,\partial/\partial\gamma_1, \ldots, \partial/\partial\gamma_q,)^T$ is the gradient operator. By observing that
$$
D_{\theta}\log\{\alpha_i(s)\}=\left(\frac{\alpha'\{g_i^{-1}(s)\}}{\alpha\{g_i^{-1}(s)\}}\sigma_{i}^{-1}X_{i}^{T},\left[\frac{\alpha'\{g_i^{-1}(s)\}}{\alpha\{g_i^{-1}(s)\}}g_i^{-1}(s)+1\right]Z_{i}^{T}\right)^{T}
$$
we rewrite the score function as
$$
\sum_{i=1}^n\int_{-\infty}^{\infty}\left[\frac{\alpha'(u)}{\alpha(u)}\sigma_{i}^{-1}X_{i}^{T},\left\{\frac{\alpha'(u)}{\alpha(u)}u+1\right\}Z_{i}^{T}\right]^{T}\{dN_{i}^*(u)-Y_{i}^*(u)\alpha(u)du\}.
$$

To arrive at an operational estimation procedure, we modify this score
function as follows. Firstly, we substitute the quantities $\alpha'(u) / \alpha(u)$ and $\alpha'(u) u / \alpha(u) +1$ with known deterministic functions which we denote by $\rho_{\beta}(u)$ and $\rho_{\gamma}(u)$, respectively. Secondly, we replace
$\alpha(u)du$ by $d\hat{A}_{n}(u,\theta)$. Thirdly, we truncate
integration at an upper limit $\tau$, where there is still a positive
probability of being at risk.  Doing so, we arrive at the estimating
equations
\begin{align*}
\Psi_{n}(\theta)= \frac{1}{n}\sum_{i=1}^n\int_{-\infty}^{\tau}\left\{\rho_{\beta}(u)\sigma_{i}^{-1}X_{i}^{T},\rho_{\gamma}(u)Z_{i}^{T}\right\}^{T}\{dN_{i}^*(u)-Y_{i}^*(u)d\hat{A}_{n}(u,\theta)\},
\end{align*}
which, for ease of exposition in the developments to follow, is rewritten as
$$
\Psi_{n}\{\theta,\hat{\eta}_{n}(\cdot,\theta)\}= \frac{1}{n}\sum_{i=1}^{n}I(\varepsilon_{\theta i}\leq\tau)\rho(\varepsilon_{\theta i})\{(\sigma_{i}^{-1}X_{i}^{T},Z_{i}^{T})-\hat{\eta}_{n}(\varepsilon_{\theta i},\theta)\}^{T}\Delta_{i},
$$
where $\rho(u)$ is a  $(p+q)\times (p+q)$ diagonal matrix with diagonal elements given by $\rho_{\beta}(u)$ repeated $p$ times followed by $\rho_{\gamma}(u)$ repeated $q$ times, and
\begin{align*}
\varepsilon_{\theta i}&=\sigma_{i}^{-1}(\log \tilde{T}_{i}-\mu_{i}),\\
\hat{\eta}_{n}(u,\theta)&=\{\hat{\eta}_{n}^{\beta}(u,\theta),\hat{\eta}_{n}^{\gamma}(u,\theta)\},\\
\hat\eta_n^{\beta}(u,\theta)&=\frac{\sum_{j=1}^nY_j^*(u)\sigma_{j}^{-1}X_j^{T}}{\sum_{j=1}^nY_j^*(u)},\\
\hat\eta_n^{\gamma}(u,\theta)&=\frac{\sum_{j=1}^nY_j^*(u)Z_j^{T}}{\sum_{j=1}^nY_j^*(u)}.
\end{align*}

The resulting estimate $\hat{\theta}_{n}$ of the true parameter value
$\theta_0$ is obtained as a minimizer of
$\|\Psi_{n}\{\theta,\hat{\eta}_{n}(\cdot,\theta)\}\|$ which in turn
enables estimation of the cumulative hazard $A_0(\cdot)$ by
$\hat{A}_{n}(\hat{\theta}_{n},\cdot)$.

\subsection{Weight functions\label{sec:weight}}

From the above we see that the weight functions in the efficient score function obey the relationship $\rho_{\gamma}(u)=u\rho_{\beta}(u)+1$. Accordingly, we suggest using using weights of the form $\rho_{\beta}(u)=\rho(u)$ and $\rho_{\gamma}(u)=\rho(u)u+1$ in practice as this choice mimics the efficient structure. As for the specific choice of $\rho(u)$, throughout the literature various authors have found rank-based estimation procedures, and associated variance estimators, to be quite insensitive to the choice of weight function \citep{linetal:1998,chenjewell:2001,jinetal:2003} and, in particular, typically suggest the use of (i) the log-rank weight, $\rho(u) = 1$, which assigns equal weight to all observations and is efficient when $e \sim$ Extreme Value (i.e., $T \sim$ Weibull), or (ii) the Gehan weight, $\rho(u) = \sum_{j=1}^nY_j^*(u)/n$, which is somewhat more data-driven in that it assigns less weight to observations for which there is less information (i.e., those corresponding to survival times in the tail of the distribution).

Alternatively, a theoretically semiparametrically efficient procedure could be based on adaptively estimating $\rho(u)$ directly from the data, perhaps using kernel smoothing \citep{tsiatis:1990, laiandying:1991, zen08eff}. However, this step introduces additional complexity beyond the use of a deterministic weight function, which can introduce some instability into the numerical estimation procedure, and, moreover, one must then consider the selection of an optimal bandwidth -- for which there are no clear guidelines in this context, and to which the results (particularly the variance estimators) can be sensitive \citep{zen08eff}. Furthermore, the resulting efficiency gain is not large in practice (cf.~\citet{chenjewell:2001}, \citet{jinetal:2003}, and  \citet{zen08eff}). For these reasons we propose the use of rank-based procedures within our semiparametric multi-parameter regression setting, and investigate some choices of weight function in Section \ref{sec:sim} and the Online Supporting Information.

\section{Asymptotic properties\label{sec:asymptotic}}

\subsection{Key results}
We show that $\hat{\theta}_{n}$ is consistent, that
$n^{1/2}(\hat{\theta}_{n}-\theta_0)$ converges to a zero mean Gaussian
distribution, and that
$n^{1/2}\{\hat{A}_{n}(\hat{\theta}_{n},\cdot)-A_0(\cdot)\}$ converges
to a tight zero mean Gaussian process. Regularity conditions and
proofs, extending the arguements of \citet{nan09} to the multi-parameter regression setting, can be found in the Online Supporting Information.

First we turn to the consistency. For this purpose let $\Psi(\theta,\eta)$ denote the limit of $\Psi_{n}(\theta,\eta)$ and let $\eta_{0}(\cdot,\theta)$ denote the limit of $\hat{\eta}_{n}(\cdot,\theta)$. Then we have the following result.

\begin{theorem}\label{thm:1} Assume that $\theta_0\in\Theta$ is the
  unique solution of $\Psi\{\theta,\eta_0(\cdot,\theta)\}=0$.  Then an
  approximate root $\hat{\theta}_n$ satisfying
  $\Psi_n\{\hat{\theta}_n,\hat{\eta}_n(\cdot,\hat{\theta}_n)\}=o_{P^*}(1)$
  is consistent for $\theta_0$.
\end{theorem}

Next, for detailing the weak convergence of $n^{1/2}(\hat{\theta}_{n}-\theta_0)$, we adopt the following notation. Let $O_{i}=(\log\tilde{T}_{i},\Delta_{i},X_i,Z_i)$ denote what we observe on the $i$th individual, and
\begin{align*}
\varepsilon_{\theta}(O)=\exp(\gamma^{T}Z)(\log \tilde{T}+\beta^{T}X)
\end{align*}
so that $\varepsilon_{\theta}(O_{i})=\varepsilon_{\theta i}$. In line with this, we shall use the short notation $\varepsilon_{\theta}$ for $\varepsilon_{\theta}(O)$ and also denote $\varepsilon_{\theta_0}$ by $\varepsilon_{0}$. Moreover we define
$$
\psi(O;\theta,\eta)=I(\epsilon_{\theta}\leq\tau)\rho(\varepsilon_{\theta})[\{\exp(\gamma^{T}Z)X^{T},Z^{T}\}-\eta(\varepsilon_{\theta},\theta)]^{T}\Delta
$$
so that
$\Psi_{n}(\theta,\eta)=\frac{1}{n}\sum_{i=1}^{n}\psi(O_i;\theta,\eta)$, and
\begin{align*}
  J(O;\theta,\eta,A) &=\psi\{O;\theta,\eta(\cdot,\theta)\} \nonumber\\
                     &\phantom{=}\;-\int_{-\infty}^{\tau} \rho(t)I(\varepsilon_{\theta}\geq t)[\{\exp(\gamma^{T}Z)X^{T},Z^{T}\}-\eta(t,\theta)]^{T}dA(t).
\end{align*}
We also define the $(p+q)\times (p+q)$ matrices  $\dot{\eta}_{\theta}(\varepsilon_{\theta},\theta)=D_{\theta} \eta(\varepsilon_{\theta},\theta)$ and $\dot{\Psi}_{\theta}\{\theta_0,\eta_0(\cdot,\theta_0)\} = D_\theta \Psi\{\theta_0,\eta_0(\cdot,\theta_0)\}$ where the $\theta$ subscript in $\dot{\eta}_{\theta}$ and $\dot{\Psi}_{\theta}$ serves as a reminder that these derivatives are taken with respect to $\theta$.

Finally, we adopt the usual empirical process notation
\begin{align*}
&Pf=\int f(o)dP(o),\\
&\Pn f=\frac{1}{n}\sum_{i=1}^{n}f(O_{i}),\\
&\Gn f=n^{1/2}(\Pn f-Pf),
\end{align*}
where $f$ denotes some bounded function on the sample space.
\begin{theorem}\label{thm:2}
  Let $\hat{\theta}_n$ be an approximate root satisfying
  $\Psi_n\{\hat{\theta}_n,\hat{\eta}(\cdot,\hat{\theta}_n)\}=o_{P^*}(n^{-1/2})$. Suppose
  that $\theta\mapsto\eta_0(\varepsilon_{\theta},\theta)$ is differentiable with
  uniformly bounded and continuous derivative
  $\dot{\eta}_{\theta}(\varepsilon_{\theta},\theta)$. Then if
  $\dot{\Psi}_{\theta}\{\theta_0,\eta_0(\cdot,\theta_0)\}$ is non-singular,
\begin{align*}
n^{1/2}(\hat{\theta}_{n}-\theta_0)=-\dot{\Psi}_{\theta}^{-1}\{\theta_0,\eta_0(\cdot,\theta_0)\} \times\Gn J(\theta_0,\eta_0,A_0) +o_{P^*}(1).
\end{align*}

\end{theorem}

Now we turn to the weak convergence of of $n^{1/2}\{\hat{A}_{n}(\hat{\theta}_{n},\cdot)-A_0(\cdot)\}$. For this we define
$$
\phi(t,\theta)=P\left\{I(\varepsilon_{\theta}\leq t)\Delta d^{(0)}(\varepsilon_{\theta},\theta)^{-1}\right\},
$$
where $d^{(0)}(t,\theta)=PI(\varepsilon_{\theta}\geq t)$, and, furthermore, the associated vector of derivatives $\dot{\phi}_{\theta} = D_\theta \phi$. We further define
\begin{align*}
H(O;t,\theta,D^{(0)},A) = I(\varepsilon_{\theta}\leq t)\Delta D^{(0)}(\varepsilon_{\theta},\theta)^{-1} -\int_{-\infty}^{t}I(\varepsilon_{\theta}\geq s)D^{(0)}(s,\theta)^{-1}dA(s).
\end{align*}
With this notation we have the following result
\begin{theorem}\label{thm:3}
  Let $\hat{\theta}_n$ be an estimator of $\theta_{0}$ such that
  $n^{1/2}(\hat{\theta}_n-\theta_{0})$ converges weakly to a zero mean
  normal distribution. Then
  $n^{1/2}\{\hat{A}_{n}(\hat{\theta}_{n},\cdot)-A_0(\cdot)\}$ converges
  weakly to a tight zero mean Gaussian process on $]-\infty,\tau]$ and
  the following holds
\begin{align*}
n^{1/2}\{\hat{A}_{n}(\hat{\theta}_{n},t)-A(t)\}=\dot{\phi}_{\theta}(t,\theta_{0})n^{1/2}(\hat{\theta}_{n}-\theta_{0}) +\Gn H(t,\theta_0,d^{(0)},A_0) +o_{P^*}(1).
\end{align*}

\end{theorem}

\subsection{A resampling procedure for estimating  asymptotic variances\label{lsvar}}

We note that the limiting covariance matrix for $n^{1/2}(\hat{\theta}_{n}-\theta_0)$ can be estimated by
\begin{align*}
\frac{1}{n}\sum_{i=1}^n (\hat{\dot{\Psi}}_{\theta}^{-1}  \hat J_i)^{\bigotimes2}
\end{align*}
where $a^{\bigotimes2} = a a^T$, and $\hat{\dot{\Psi}}_{\theta}$ and
$\hat J_i= J(O_i;\hat\theta_n,\hat\eta_n,\hat A_n)$ are estimates of their
theoretical counterparts. While $\hat J_i$ is easily computed,
estimation of $\dot{\Psi}_{\theta}$ would require an estimate of the
error hazard $\alpha(\cdot)$ which is difficult to obtain
reliably. The classical solution to this problem is to produce a
sample $\hat{\theta}^{b}$, $b=1,\ldots m$, by solving perturbed
estimating equations (often based on \citet{parzenetal:1994}), from
which the limiting covariance matrix can be estimated directly; such
procedures are, however, computationally intensive owing to solving
estimating equations multiple times.

As the representation in Theorem \ref{thm:2} is of the form considered
in \citet{zeng08}, we may apply their more modern resampling approach
which requires re-evaluating (but not re-solving) estimating
equations. Their approach is based on the fact that
$n^{-1/2} \Psi_n(\hat\theta_n+n^{-1/2}G) = \dot \Psi^{-1} G +
o_{P^*}(1)$, where $G$ is a zero-mean Gaussian $(p+q)$-vector
independent of the data, which motivates the least squares estimate
$\hat{\dot{\Psi}}_{\theta} = (M^TM)^{-1} M^T U$ where $M$ and $U$ are
matrices whose $b$th rows are, respectively, given by the zero-mean
Gaussian vector $G^{b}$, and the vector
$n^{-1/2}\Psi_n(\hat\theta_n+n^{-1/2}G^b)$, $b=1,\ldots,m$.

In a similar manner to the regression coefficients, we can estimate
the limiting variance of
$n^{1/2}\{\hat{A}_{n}(\hat{\theta}_{n},t)-A_0(t)\}$ using
\begin{align*}
\frac{1}{n}\sum_{i=1}^n \{- \hat{\dot{\phi}}_{\theta}(t)\hat{\dot{\Psi}}_{\theta}^{-1} \hat J_i + \hat H_i(t)\}^2
\end{align*}
where $\hat{\dot{\Psi}}_{\theta}$ is estimated using least squares as
described in the previous paragraph, and
$\hat{H}_i(t) = H(O_i; t,\hat\theta_n,D_n^{(0)},\hat A_n)$ where
$D_n^{(0)}(t,\theta)=\Pn I(\varepsilon_{\theta}\geq t) = \frac{1}{n}
\sum_{i=1}^n Y_i^*(t)$. However, an estimator
$\hat{\dot{\phi}}_{\theta}(t)$ of $\dot{\phi}_{\theta}(\theta_{0},t)$
is difficult to obtain directly. Instead we adapt the resampling idea
of \citet{zeng08} to obtain an estimator of
$\dot{\phi}_{\theta}(\theta_{0},t)$.  In particular, according to the
asymptotic representation of
$n^{1/2}\{\hat{A}_{n}(\hat{\theta}_{n},\cdot)-A_0(\cdot)\}$ in Theorem
\ref{thm:3}, we find that
\begin{align*}
&n^{1/2}\{\hat{A}_{n}(\hat\theta_n+n^{-1/2}G,t) - \hat{A}_{n}(\hat{\theta}_{n},t)\} \\
\qquad&=
n^{1/2}\{\hat{A}_{n}(\hat\theta_n+n^{-1/2}G,t)-A_0(t)\} - n^{1/2}\{\hat{A}_{n}(\hat{\theta}_{n},t)-A_0(t)\} \\
\qquad&=\dot{\phi}_{\theta}(t,\theta_{0}) G +o_{P^*}(1)
\end{align*}
which motivates the least squares estimate
$\hat{\dot{\phi}}_{\theta}(t) = \{(M^T M)^{-1} M^T \tilde U\}^T$ where
$M$ is as before, and $\tilde U$ is a matrix whose $b$th row is given
by
$n^{1/2}\{\hat{A}_{n}(\hat\theta_n+n^{-1/2}G^b,t) -
\hat{A}_{n}(\hat{\theta}_{n},t)\}$, $b=1,\ldots,m$.

With the estimates of the limiting variances of $\hat\theta_n$ and
$\hat{A}_{n}(\hat{\theta}_{n},t)$, we can straightforwardly produce
Wald-type confidence intervals for the parameters, and confidence bands
for the error cumulative hazard. In the latter case, as is standard,
it is preferable to produce confidence bands on the $\log A_0(t)$ scale
first and back-transform to the $A_0(t)$ scale. For functionals of
$\theta_0$ and $A_0(t)$, one could apply the functional delta method
\citep{andersen93}. However, in line with the resampling approaches
discussed above, we suggest the use of the conditional multiplier
method from empirical process theory \citep{vdvaartwellner96} which is
described in the Online Supporting Information.

\section{Numerical studies}

\subsection{Simulation\label{sec:sim}}

We now investigate the performance of our procedure in finite samples by way of a simulation study. In particular, we generated survival times according to the following setup: $\mu = -(X_1 \beta_1 + X_2 \beta_2)$ and $\log \sigma = - X_1 \gamma_1$ with covariates $X_1 \sim \text{Bernoulli}(0.5)$ and $X_2 \sim \text{Uniform}(0,1)$,  parameter vector $\theta = (\beta_1,\beta_2,\gamma_1) = (1,1,1)$, and error distribution $e \sim N(0,1)$. Furthermore, we considered a sample size of 100 where survival times were randomly censored according to a log-normal distribution with unit-scale and location set to achieve censored proportions of approximately 20\% or 50\% respectively.

Since the estimation equations are non-smooth step-functions of the
parameters, we applied the Nelder-Mead optimisation procedure as
implemented in the \texttt{optim} function in the \texttt{R}
programming language. It is also worth highlighting the fact that the
estimation equations, as we have presented them, depend on a
threshold, $\tau$, which was required for our asymptotic derivations
to ensure a non-empty risk set so that denominators are theoretically
bounded away from zero. To investigate the sensitivity of our approach
to the inclusion of $\tau$, we consider thresholds of 2 and
$\infty$. As for the choice of weight function we consider the log-rank weight, $\rho(u) = 1$, and the normal weight, $\rho(u) = f_e(u)/(1-F_e(u)) - u$ where $f_e$ and $F_e$ are the normal pdf and cdf functions, which is the true, efficient weight in our simulation study.

In total we present eight simulation scenarios comprising one sample
size, two censored proportions, two threshold values, and two choices of weight function. Each of
these scenarios was replicated 5000 times. Within each replicate we
estimated the following quantities: (i) the parameter vector,
$\theta$, (ii) the cumulative error hazard at the median error,
$A(0) =$ 0.6931, (iii) the conditional survivor function for
covariate profile $x^{(1)}= (x^{(1)}_1, x^{(1)}_2)^T = (1,1)^T$
evaluated at the median time for this covariate profile,
$S(t^{(1)}_{0.5}\,|\, x^{(1)})=$0.5, and (iv) the ratio of
the median for $x^{(1)}=(1,1)^T$ to the median for
$x^{(2)} = (0,1)^T$, $r(x^{(1)},x^{(2)})=$ 0.3679. For quantities
(i) and (ii), Wald-type confidence intervals were produced where the
limiting variances were estimated using least squares (as described in
Section \ref{lsvar}) with $m=1000$, and for quantites (iii) and (iv)
the conditional multiplier method was used (as described in the
Online Supporting Information) with $m=1000$.

\begin{table}
\def~{\hphantom{0}}
\caption{Results of simulation study}
\begin{tabular}{cccrrrrc@{\quad}rrrr}
\hline
 &&               & \multicolumn{4}{c}{Log-rank} &&  \multicolumn{4}{c}{Normal (true, efficient)} \\
\cline{4-7}\cline{9-12}
 $\tau$ & Cens. &  Parameter  & Bias & SE & SEE & Cov. && Bias &  SE & SEE & Cov. \\
 \hline
 $2$ & 20\%
       & $\beta_1$     &   $<$0.001 & 0.099 & 0.099 & 94.8  && -0.001 & 0.097 & 0.096 & 94.4    \\
      && $\beta_2$     &   0.008    & 0.189 & 0.185 & 94.0  &&  0.002 & 0.178 & 0.169 & 93.1    \\
      && $\gamma_1$    &  -0.001    & 0.174 & 0.179 & 94.1  &&  0.004 & 0.176 & 0.157 & 91.8    \\
      && $A$           &  -0.009    & 0.168 & 0.156 & 94.5  && -0.002 & 0.164 & 0.153 & 94.4    \\
      && $S$           &   $<$0.001 & ---         & ---         & 94.6  &&$<$0.001& ---         & ---         & 94.3    \\
      && $r$           &   0.003    & ---         & ---         & 95.1  &&  0.004 & ---         & ---         & 94.8    \\
$2$ & 50\%
       & $\beta_1$     &   $<$0.001 & 0.127 & 0.124 & 94.2  && -0.006 & 0.120 & 0.121 & 94.8    \\
      && $\beta_2$     &   0.005    & 0.216 & 0.210 & 93.3  &&  0.003 & 0.201 & 0.193 & 93.7    \\
      && $\gamma_1$    &  -0.015    & 0.208 & 0.230 & 94.6  &&  0.004 & 0.215 & 0.206 & 93.2    \\
      && $A$           &  -0.010    & 0.219 & 0.203 & 95.3  && -0.003 & 0.217 & 0.201 & 94.8    \\
      && $S$           &   0.008    & ---         & ---         & 95.3  &&  0.005 & ---         & ---         & 95.5    \\
      && $r$           &   0.006    & ---         & ---         & 95.5  &&  0.008 & ---         & ---         & 95.9    \\
$\infty$ & 20\%
       & $\beta_1$     &   $<$0.001 & 0.100 & 0.099 & 94.3  && -0.001 & 0.098 & 0.096 & 94.3    \\
      && $\beta_2$     &   0.003    & 0.189 & 0.183 & 93.6  && -0.002 & 0.179 & 0.169 & 93.4    \\
      && $\gamma_1$    &  -0.007    & 0.175 & 0.177 & 93.4  &&  0.003 & 0.171 & 0.156 & 93.2    \\
      && $A$           &  -0.004    & 0.165 & 0.156 & 94.3  && -0.005 & 0.165 & 0.152 & 94.2    \\
      && $S$           &   0.001    & ---         & ---         & 94.5  &&  0.004 & ---         & ---         & 94.4    \\
      && $r$           &   0.004    & ---         & ---         & 95.3  &&  0.005 & ---         & ---         & 94.6    \\
$\infty$ & 50\%
       & $\beta_1$     &   0.003    & 0.126 & 0.124 & 94.5  &&$<$0.001& 0.124 & 0.121 & 94.1    \\
      && $\beta_2$     &   0.001    & 0.215 & 0.208 & 93.1  && -0.001 & 0.207 & 0.191 & 92.6    \\
      && $\gamma_1$    &  -0.016    & 0.215 & 0.230 & 93.7  &&  0.008 & 0.224 & 0.206 & 92.0     \\
      && $A$           &  -0.010    & 0.230 & 0.204 & 95.1  && -0.008 & 0.223 & 0.201 & 94.5    \\
      && $S$           &   0.009    & ---         & ---         & 95.3  &&  0.005 & ---         & ---         & 94.4    \\
      && $r$           &   0.005    & ---         & ---         & 95.4  &&  0.007 & ---         & ---         & 95.3\\
\hline
\end{tabular}
\label{tab:sim}

{\footnotesize
Cens., censored proportion; Bias, median bias; SE, standard error of estimates; SEE, median of estimated standard error, Cov., empirical coverage percentage for 95\% confidence interval; $A = A(0)$; $S=S(t^{(1)}_{0.5}\,|\, x^{(1)})$; $r = r(x^{(1)},x^{(2)})$. Since the variances for the functionals $S$ and $r$ are not estimated directly within our scheme, SE and SEE are not shown in those cases.}
\end{table}

The results are summarised in Table \ref{tab:sim}. It is clear that the estimates are reasonably unbiased in all cases and the associated 95\% confidence intervals achieve a coverage percentage which is close to the desired nominal level for both choices of weight function, and, in either case, the results for $\tau=2$ and $\tau=\infty$ are very similar. The estimated standard errors capture the true variations adequately, and, moreover, the efficiency based on the log-rank weights is very close to that of the true, efficient weights (which is in line with the findings of other authors in the simpler accelerated failure time model context). Additional simulation results are given in the Online Supporting Information which cover $n=50$ and $n=500$, and Gehan weights; the results are comparable with those shown here.

\subsection{Lung cancer data}

We now apply our model to data arising from a lung cancer study which was the subject of a 1995 Queen's University Belfast PhD thesis by P. Wilkinson (previously analysed in \citet{burmac:2017}). This observational study pertains to 855 individuals who were diagnosed with lung cancer during the one-year period 1st October 1991 to 30th September 1992, and these individuals were followed up until 30th May 1993 (approximately 20\% of survival times were right-censored). The primary interest was to investigate the differences between the following treatment groups: palliative care, surgery, chemotherapy, radiotherapy, and a combined treatment of chemotherapy and radiotherapy. While various other covariates were measured (see \citet{burmac:2017}), the aim here is to illustrate our semiparametric multi-parameter regression methodology for the treatment model.

The results of the fitted model are given in Table \ref{tab:lung} where the log-rank weights were used in the estimation procedure (see the Online Supporting Information for other choices of weights which yield numerically very similar results). Firstly note that the $\beta$ coefficients are all negative, and statistically significant, suggesting an improvement in survival relative to palliative care group. The $\gamma$ coefficients of radiotherapy and the combined treatment of both chemotherapy and radiotherapy differ statistically from zero, indicating that the quantile ratios are non-constant. Furthermore, note that the $\gamma$ coefficients are positive and thus the quantile ratios decrease over the timeframe, i.e., the effectiveness of these treatments diminishes over time. Recall that, this being an observational study, the estimated effects are not ``treatment effects'' in the sense of a randomised trial, but, notwithstanding this, analyses of observational effects are still useful in their own right.

\begin{table}
\def~{\hphantom{0}}
\caption{Regression coefficients for model fitted to lung cancer data}
\begin{tabular}{lcrrrrrrr}
\hline
 && \multicolumn{3}{c}{Location} & \multicolumn{3}{c}{Scale} & Joint \\
\hline
Treatment Group &Sample size &  Est. & SE   & P-val    & Est. & SE   & P-val   & P-val \\
 Palliative care & 441& 0.00 & ---  & ---     & 0.00 & ---  & ---  & --- \\
Surgery         & 79&-2.65 & 0.17 & $<$0.01 & 0.31 & 0.20 & 0.12    & $<$0.01 \\
Chemotherapy    &45& -0.54 & 0.27 & 0.04    & 0.04 & 0.11 & 0.73    & 0.06 \\
Radiotherapy    &256& -1.08 & 0.10 & $<$0.01 & 0.30 & 0.07 & $<$0.01 & $<$0.01 \\
Chemo. \& radio &34& -1.87 & 0.12 & $<$0.01 & 0.94 & 0.17 & $<$0.01 & $<$0.01\\
\hline
 \end{tabular}
\label{tab:lung}

{\footnotesize
Est., estimated location ($\beta$) or scale ($\gamma$) coefficient; SE, standard error; P-val, p-value. The joint p-value corresponds to testing that $\beta_j = \gamma_j = 0$.}
\end{table}

It is also of interest to test whether the overall effect of a given treatment is statistically significant, i.e., testing $\beta_j = \gamma_j = 0$ for the $j$th group. Asymptotic normality of the estimated parameter vector means that this can be achieved by comparing $(\hat\beta_j, \hat\gamma_j) \hat\Sigma^{-1}_{\beta_j, \gamma_j} (\hat\beta_j, \hat\gamma_j)^T$ to a $\chi^2_2$ distribution where $\Sigma_{\hat\beta_j, \hat\gamma_j}$ is the $2\times2$ covariance matrix for the pair $(\hat\beta_j, \hat\gamma_j)$. The resulting p-values for this test are shown in the last column of Table \ref{tab:lung}.

\begin{figure}[!htbp]
\includegraphics[width=0.95\textwidth, trim = 0.0cm 0.4cm 0.5cm 0.8cm, clip]{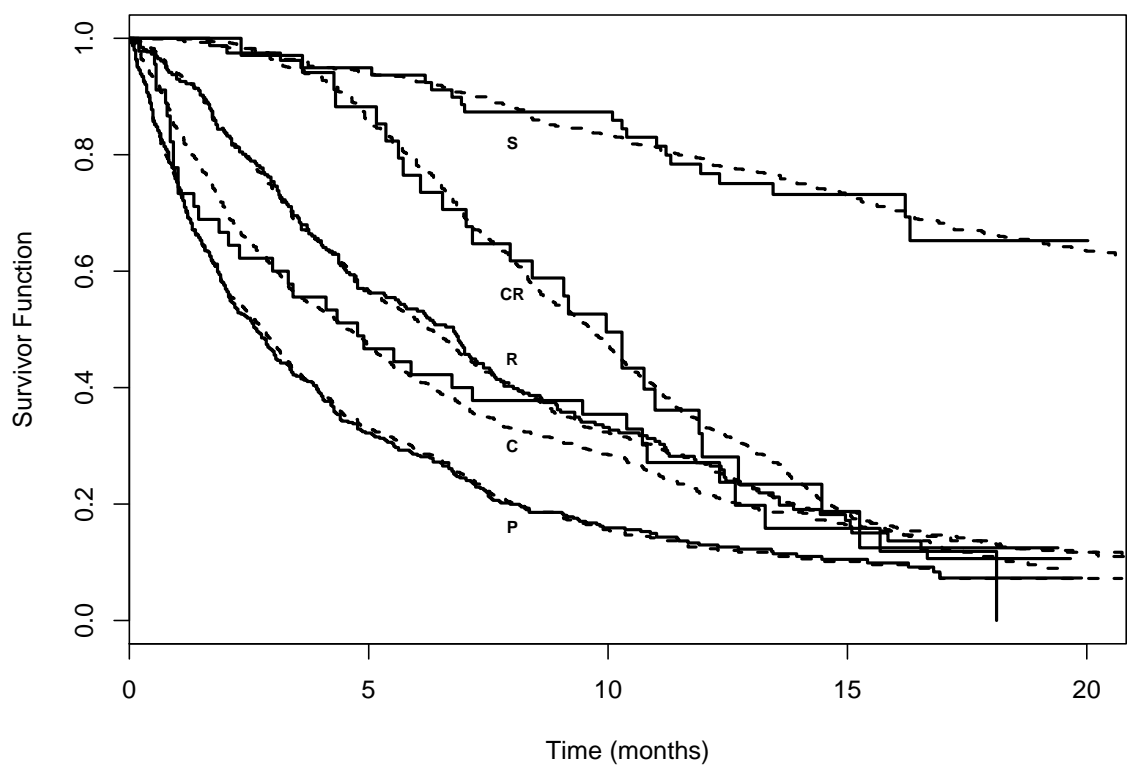}
\caption{Kaplan-Meier (solid) curves with model-based curves (dash) overlayed where
P $=$ palliative, C $=$ chemotherapy, R $=$ radiotherapy, CR $=$ chemotherapy \& radiotherapy combined, and S $=$ surgery, respectively.}\label{fig:KMfit}
\end{figure}

Figure \ref{fig:KMfit} compares the fitted survivor curves to the
Kaplan-Meier curves.  We can see that the model provides an excellent fit to the data. Furthermore, we see that, relative to
the palliative care curve, the various curves have different shapes,
particularly those of radiotherapy and the combined treatment. Indeed,
the curves converge at a rate which is indicative of a reduction in
treatment effectiveness over time, and is something which cannot be
handled by a basic accelerated failure time model. This highlights the
flexibility of the multi-parameter regression extension wherein the
scale parameter depends on covariates thereby facilitating survivor
curves corresponding to non-constant quantile ratios.

While the results of Table \ref{tab:lung} provide useful information
on the nature of the treatment effects, and which are statistically
significant, we now consider the quantile ratios which allow us to
quantify the effect of each treatment on lifetime. From Figure \ref{fig:QR}, we can see that the population assigned to surgery has the highest survival, and the difference is sustained with
time. Even though it may appear to diminish with time, it is clear
that a horizontal line corresponding to a constant quantile ratio
easily fits within the confidence bands. This is in line with the
non-significant scale coefficient seen in Table \ref{tab:lung}. The
combined treatment of chemotherapy and radiotherapy is particularly
effective early on but drops sharply in effectiveness over
time. Radiotherapy provides a more modest improvement in lifetime but
has a similar performance to the combined treatment later in
time. Finally, chemotherapy appears to have a relatively weak effect
over the whole lifetime.

\begin{figure}[!htbp]
\includegraphics[width=0.95\textwidth, trim = 0.0cm 0.4cm 0.5cm 0.8cm, clip]{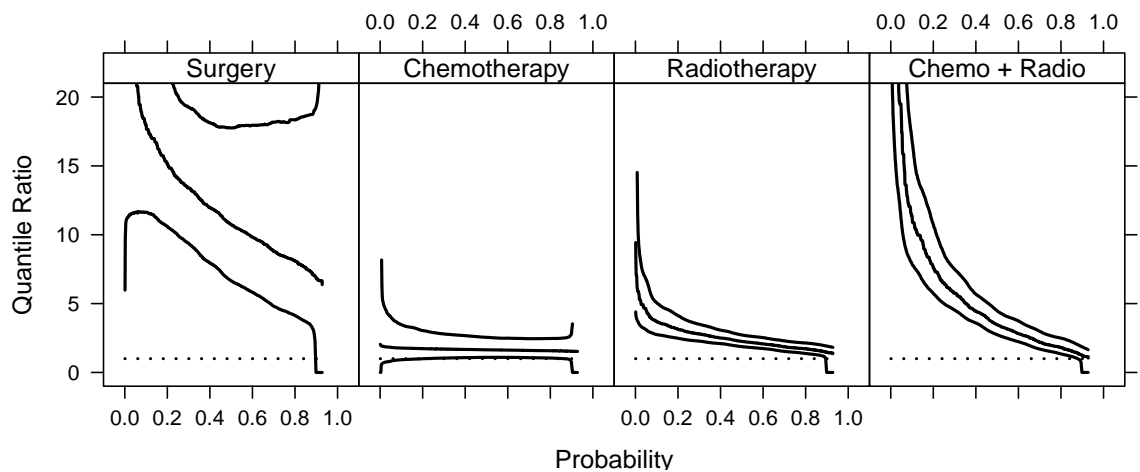}
\caption{Quantile ratios  (relative to palliative care)  for each treatment group relative to palliative care along with pointwise 95\% confidence intervals (solid). Reference line at unity (dot) also shown.\label{fig:QR}}
\end{figure}

\section{Discussion \label{discuss}}

In this paper we have extended the semiparametric accelerated failure time model to multi-parameter regression status by jointly modelling the location and scale parameters of its log-linear model representation. This brings together the structural flexibility of multi-parameter regression modelling and the robustness of an unspecified baseline hazard. The resulting model can be interpreted on the lifetime scale through its quantile ratio. The form of the quantile ratio directly generalizes that of the accelerated failure time model in that it depends on the quantile in question (rather than being constant over all quantiles) which allows for effects which change over the individual lifetime. Moreover, this modelling framework produces a new semiparametric test of the accelerated lifetime property for a given covariate which is adjusted for other covariates in the model.

Clearly, the combination of multi-parameter regression and semiparametric modelling is  fruitful, and some approaches in this direction exist in current literature, e.g., the  proportional-hazards-accelerated-failure-time hybrid of \citet{chenjewell:2001} and the proportional-hazards-Aalen hybrid of \citet{scheikezhang:2002, scheikezhang:2003}. However, in the aforementioned models, the two regression components essentially both correspond to distributional scale-type parameters, i.e., the components play similar roles. Furthermore, the models do not have a natural scale for interpretation which is related to the previous point. In contrast, the choice of jointly modelling location and dispersion of the error distribution (or scale and shape of the survival distribution) provides somewhat more ``orthogonal'' components for the inclusion of covariates -- heteroscedastic linear models are familiar in other areas such as econometrics -- and, as mentioned above, yields an interpretation on the lifetime scale. Note also that the Aalen model \citep{aalen:1980}, while not related to multi-parameter regression, is another flexible and robust survival regression model, being fully non-parametric, but which can be somewhat difficult to interpret (owing to the regression functions being on a cumulative hazard scale); \citet{martpipper:2013,martpipper:2014} considered its interpretation through the ``odds of concordance''.

On model interpretability, we could even go as far as criticizing hazard-based models in general, as did D.R. Cox when he stated that ``accelerated life models are in many ways more appealing [than hazard modelling] because of their quite direct physical interpretation'' \citep{reid:1994}, i.e., such models are interpretable on the lifetime scale as is the case for our model. While we might expect the basic accelerated failure time to be more popular on the basis of this ``direct physical interpretation'', inference in semiparametric accelerated failure time models has, historically, been comparatively more difficult than in the Cox model, which we discuss in the following paragraph.

Estimation of parameters and baseline cumulative hazard function for our model is based on the creation of a time-transformed counting process (see Section \ref{sec:esteq}), yielding a set of estimation equations which directly generalize those of the basic accelerated failure time model. It is noteworthy that our counting process formulation means that the estimating equations extend immediately to more general settings such as multiple events and time-varying covariates. Unlike those of the Cox model, however, the estimation equations for accelerated failure time models (and, hence, for our model) are such that the resulting covariance matrix for the estimators has a non-analytic form. However, we have overcome this by making use of modern empirical process theory in combination with a modification of a resampling procedure due to \citet{zeng08} (as described in Section \ref{sec:asymptotic} and the Online Supplementary Information). In particular, our proposal does not require resolving of estimation equations, and permits straightforward inference for any (conditional) survival functional of interest.

In summary, the semiparametric multi-parameter regression model of this paper achieves flexibility through its basic model structure, robustness with respect to distributional assumptions as its baseline distribution is unspecified, and is interpretable on the lifetime scale similar to that of the accelerated failure time model which it directly extends. Our inferential framework combines modern approaches in a novel way which is useful for rank-based estimation in general (beyond our setting and the accelerated failure time model which we generalize), for example, those used in the estimation of transformation models.

\section*{Acknowledgement}
This research was supported by the Irish Research Council (New Foundations Award) and the Royal Irish Academy (Charlemont Award).

\bibliographystyle{Chicago}
\bibliography{refs}

\newpage
\appendix

\centerline{\sc{Appendix}}
\medskip

\setcounter{thm}{0}

\section{Assumptions}

\begin{assumption}\label{ass:0}
  The parameter $\theta_0=(\beta_0^T,\gamma_0^T)^T$ lies in the interior of a compact set $\mathcal{B}\times\Gamma=\Theta\subset
  \R^{p+q}$.
\end{assumption}

\begin{assumption}\label{ass:1}
There exists constants
  $\tau<\infty$ and $\epsilon$, such that
  $pr(\varepsilon_{\theta}\geq \tau)\geq\epsilon>0$ for all $X$, $Z$
  and $\theta\in\Theta$.
\end{assumption}

\begin{assumption}\label{ass:2}
The covariates $X$ and $Z$ are uniformly bounded with probability one.
\end{assumption}

\begin{assumption}\label{ass:3}
  The error density $f$ and its derivative $f'$ are bounded and
  $\int \{f'(t)/f(t)\}^2f(t)dt<\infty$.
\end{assumption}

\begin{assumption}\label{ass:4}
$\log(C)$ has uniformly bounded densities.
\end{assumption}

\begin{assumption}\label{ass:5}
  The diagonal of $\rho(\varepsilon_{\theta})$ is differentiable in
  $\theta$ with bounded continuous derivative
  $\dot{\rho}_{\theta}(\varepsilon_{\theta})$ and
  $\left\{\rho(\varepsilon_{\theta}):\theta\in\Theta\right\}$ is a
  bounded Donsker class.
\end{assumption}

\begin{assumption}\label{ass:6}
 $\varepsilon_{0}$ has finite second order moment.
\end{assumption}

\begin{remark}
  Assumptions \ref{ass:3} and \ref{ass:4} are those assumed in
  \citet{nan09} directly adapted from \citet{ying93}.
\end{remark}

\section{Asymptotic results}

The asymptotic properties are established by extending the proofs of
\citet{nan09} to the multi parameter regression setting. We show asymptotic
linearity of $\Psi_n(\theta,\hat{\eta}_n)$ in $\theta$ in a
neighborhood of the true value $\theta_0$. This will rely heavily on
the fact that the function classes $\mathcal{F}_{0}$ and
$\mathcal{F}_{1}$ defined below have bracketing numbers of polynomial
order.
$$\mathcal{F}_0=\{I(\varepsilon_{\theta}\geq
t):t\in]-\infty,\tau],\theta\in\Theta\}$$ and
\begin{align*}
\mathcal{F}_{1}=\left\{ I(\varepsilon_{\theta}\geq t)(e^{\gamma^{T}Z}X^{T},Z^{T}) :t\in]-\infty,\tau],\theta\in\Theta\right\}.
\end{align*}

With $N_{[]}\{\epsilon,\mathcal{F}_{0},L_{2}(P)\}$ and $N_{[]}\{\epsilon,\mathcal{F}_{1},L_{2}(P)\}$ denoting the bracketing numbers of $\mathcal{F}_{0}$ and $\mathcal{F}_{1}$, respectively, we have the following result
\begin{lemma}\label{lemma:1}
There exist $K_{1}>0$, $K_{2}>0$ so that for all $\epsilon<1$,
\begin{align}
\label{eq:1}N_{[]}\{\epsilon,\mathcal{F}_{0},L_{2}(P)\}&\leq K_{1}\epsilon^{-3(p+q+3/2)},\\
\nonumber N_{[]}\{\epsilon,\mathcal{F}_{1},L_{2}(P)\}&\leq K_{2}\epsilon^{-(6p+7q+9)}.
\end{align}
\end{lemma}

\begin{proof}
  First note that due to the compactness of $\Theta$ and Assumption
  \ref{ass:2} there exist $m\in L_{2}(P)$ so that
  $|\varepsilon_{\theta_{1}}(o)-\varepsilon_{\theta_{2}}(o)|\leq
  m(o)\|\theta_{1}-\theta_{2}\|$. It follows as in \citet{vdvaart98}
  Example 19.7 that there exist $\tilde{K}_{1}>0$ and
  $\tilde{\delta}>0$ so that the class
  $\mathcal{G}=\{\varepsilon_{\theta}:\theta\in\Theta\}$ can be
  covered by $I$ brackets of the form
  $[f_{i}-\epsilon m,f_{i}+\epsilon m], \: f_{i}\in\mathcal{G}$ where
  $I\leq\tilde{K}_{1}\epsilon^{-(p+q)} \text{ for }
  \epsilon<\tilde{\delta}$. Now let
  $-\epsilon^{-1/2}=t_{1}<\ldots<t_{J}=\tau$ be a partition of the
  interval $[-\epsilon^{-1/2},\tau]$ so that
  $|t_{j}-t_{j+1}|\leq\epsilon$ and so that
  $J\leq 1+\tau\epsilon^{-1}+\epsilon^{-3/2}$. From this partition
  define
\begin{align*}
&l_{i,j}=f_{i}-\epsilon m-t_{j} \text{ for } j=1,\ldots, J, i=1,\ldots,I\\
&u_{i,j}=f_{i}+\epsilon m -t_{j-1} \text{ for } j=2,\ldots, J, i=1,\ldots,I\\
&u_{i,1}=\infty
\end{align*}
and note that the brackets $\{I(l_{i,j}\geq0),I(u_{i,j}\geq0)\}$ cover $\mathcal{F}_{0}$. Moreover note that from Assumptions \ref{ass:2} and \ref{ass:6} the class $\mathcal{G}$ has an $L_{2}(P)$ envelope which we shall term $F$. It now follows that
\begin{align*}
\|I(l_{i,1}\geq0)-I(u_{i,1}\geq0)\|_{2}^{2}&=\int\{1-I(l_{i,1}\geq0)\}^{2}dP\\
&=\int I(f_{i}-\epsilon m+\epsilon^{-1/2}<0)dP\\
&\leq pr(-F-\tilde{\delta}m<-\epsilon^{-1/2})\\
&\leq \|F+\tilde{\delta} m\|_{2}^{2}\epsilon
\end{align*}
for all $\epsilon\leq\tilde{\delta}$, where the last inequality
follows from direct application of Markov's generalized
inequality. Similarly for $j>1$ and for all $\tilde{C}>0$
\begin{align*}
\|I(l_{i,j}\geq0)-I(u_{i,j}\geq0)\|_{2}^{2}&=\int I(-\epsilon m+t_{j-1}\leq f_{i}\leq \epsilon m+t_{j})dP\\
&\leq\int I\{-\epsilon(m+1)\leq f_{i}-t_{j}\leq\epsilon(m+1)\}dP\\
&=\int_{\{m\geq \tilde{C}\}}I\{-\epsilon(m+1)\leq f_{i}-t_{j}\leq\epsilon (m+1)\}dP\\
&\phantom{=}\;+\int_{\{m < \tilde{C}\}}I\{-\epsilon (m+1)\leq f_{i}-t_{j}\leq\epsilon (m+1)\}dP\\
&\leq pr(m\geq \tilde{C})+pr\{|f_{i}-t_{j}|\leq\epsilon (\tilde{C}+1)\}.
\end{align*}
Using Markov's inequality we get
$pr(m\geq \tilde{C})\leq \|m\|_{2}^{2} \tilde{C}^{-2}$. Furthermore
from Assumptions \ref{ass:3} and \ref{ass:4} there exists a constant
$\tilde{K}_{2}$ such that for all $\epsilon$ and $\tilde{C}$
$$
pr\{|f_{i}-t_{j}|\leq \epsilon (\tilde{C}+1)\}\leq \tilde{K}_{2} (\tilde{C}+1)\epsilon.
$$
Consequently, by choosing $\tilde{C}=\epsilon^{-1/3}$, we obtain the following
bound for $j>1$
$$
\|I(l_{i,j}\geq0)-I(u_{i,j}\geq0)\|_{2}^{2}\leq \left(\|m\|_{2}^{2}+\tilde{K}_{2}\right)\epsilon^{2/3}+\tilde{K}_{2}\epsilon
$$
Combining all the bounds we conclude that there exists $K>0$ and
$\delta>0$ such that for all $\epsilon\leq\delta$, $N_{[]}\{K \epsilon^{1/3},\mathcal{F}_{0},L_{2}(P)\}\leq I J$. A rescaling then proves (\ref{eq:1}).

For the second part of the lemma note that $\mathcal{F}_{0}$ is
bounded as is
$\mathcal{G}_{1}=\{e^{\gamma^{T}Z}X^{T}: \gamma\in\Gamma\}$ and
$\mathcal{G}_{2}=\{Z^{T}\}$. Finally according to \citet[Example
19.7]{vdvaart98}, the bracketing number of $\mathcal{G}_{1}$ is less
than of the order $q$. Adding up we see that the bracketing number of
$\mathcal{F}_{1}$ is less than of the order
$q+3(p+q+3/2)+3(p+q+3/2)=6p+7q+9$.

\end{proof}

\begin{theorem}\label{thm:1} Assume that $\theta_0\in\Theta$ is the
unique solution of $\Psi(\theta,\eta_0(\cdot,\theta))=0$.  Then an approximate root $\hat{\theta}_n$ satisfying
$\Psi_n\{\hat{\theta}_n,\hat{\eta}_n(\cdot,\hat{\theta}_n)\}=o_{P^*}(1)$
is consistent for $\theta_0$.
\end{theorem}

\begin{proof}[(i)]

  Let $\|\cdot\|$ denote the supremum norm. Since $\theta_0$ is the
  unique solution to $\Psi\{\theta,\eta_0(\cdot,\theta)\}=0$ and
  $\Theta$ is compact it follows that for any fixed $\epsilon>0$,
  there exists a $\delta>0$ such that
  $pr\left(\|\hat{\theta}_n-\theta_0\|>\epsilon\right)\leq
  pr\left[\|\Psi\{\hat{\theta}_n,\eta_0(\cdot,\hat{\theta}_n)\}\|>\delta\right]$. If
  we can show that
\begin{align}\label{eq:2}
  \|\Psi\{\hat{\theta}_n,\eta_0(\cdot,\hat{\theta}_n)\}\|=o_{P^*}(1),
\end{align}
then the consistency of $\hat{\theta}_n$ follows.

  We first show that
  $\|\hat\eta_n(t,\theta)-\eta_0(t,\theta)\|=o_{P^*}(1)$. Define
\begin{align*}
D_{n}^{(0)}(t,\theta)&=\Pn\left\{I(\varepsilon_{\theta}\geq t)\right\},\\
D_{n}^{(1)}(t,\theta)&=\Pn \left\{I(\varepsilon_{\theta}\geq t)(e^{\gamma^{T}Z}X^{T},Z^{T})\right\},\\
d^{(0)}(t,\theta)&=P\left\{I(\varepsilon_{\theta}\geq t)\right\},\\
d^{(1)}(t,\theta)&=P\left\{I(\varepsilon_{\theta}\geq t)(e^{\gamma^{T}Z}X^{T},Z^{T})\right\}.
\end{align*}
Thus,
$\hat{\eta}_n(t,\theta)=D_{n}^{(1)}(t,\theta)/D_{n}^{(0)}(t,\theta)$
and
$\eta_0(t,\theta)=d^{(1)}(t,\theta)/d^{(0)}(t,\theta)$.

As the classes $\mathcal{F}_{0}$ and $\mathcal{F}_{1}$ are Donsker by
Lemma \ref{lemma:1} it follows that $\|D_n^{(0)}-d^{(0)}\|=o_{P^*}(1)$
and $\|D_{n}^{(1)}-d^{(1)}\|=o_{P^*}(1)$ and
$n^{1/2}\{D_{n}^{(k)}(t,\theta)-d^{(k)}(t,\theta)\}$ converge to zero
mean Gaussian processes on $]-\infty,\tau]\times\Theta$. Since
$D_n^{(0)}$ (almost surely) and $d^{(0)}$ are bounded away from zero,
\begin{align}\label{eq:3}
\|\hat\eta_n-\eta_0\|=o_{P^*}(1).
\end{align}

The random functions
$D_n^{(0)}(\varepsilon_{\theta},\theta)=n^{-1}\sum_{i=1}^n
I(\varepsilon_{\theta i}\geq \varepsilon_{\theta})$ and
$D_n^{(1)}(\varepsilon_{\theta},\theta)=n^{-1}\sum_{i=1}^n
I(\varepsilon_{\theta i}\geq \varepsilon_{\theta})(e^{\gamma^T
  Z}X^{T},Z^{T})$ can be expressed as the limit of convex combinations
of elements of the Donsker classes
$\{I(s\geq \varepsilon_{\theta}):s\in]-\infty,\tau],\theta\in\Theta\}$ and $\{I(s\geq \varepsilon_{\theta})(e^{\gamma^T
  Z}X^T,Z^T) : s\in]-\infty,\tau],\theta\in\Theta\}$ and
are bounded. Thus, they belong to the closed convex hull of those classes
which is Donsker by \citet[Theorem 2.10.3]{vdvaartwellner96}. By
Assumption \ref{ass:1}, $D_n^{(0)}$ is bounded away from zero almost
surely, so that
$\{\hat{\eta}_n(\varepsilon_{\theta},\theta):\theta\in\Theta\}$ is Donsker
by \citet[Example 2.10.9]{vdvaartwellner96}. See \citet[pp. 844-845]{kim13}.

The class of bounded functions
\begin{align*}
\left\{\psi(O;\theta,\hat\eta_n)=I(\varepsilon_{\theta}\leq\tau)\rho(\varepsilon_{\theta})\{(e^{\gamma^{T}Z}X^{T},Z^{T})-\hat{\eta}_{n}(\varepsilon_{\theta},\theta)\}^T\Delta:\theta\in\Theta\right\}
\end{align*}
is a Glivenko-Cantelli class. By adding and subtracting the same term,
and by the triangle inequality, we then have that
\begin{align*}
\|\Psi_n(\theta,\hat{\eta}_n)-\Psi(\theta,\eta_0)\|&=\|\Pn \psi(O;\theta,\hat{\eta}_n)-P \psi(O;\theta,\eta_0)\|\\
                                                   &\leq \|(\Pn-P) \psi(O;\theta,\hat{\eta}_n)\}\|\\
  &\phantom{=}\;+ \|P I(\varepsilon_{\theta}\leq\tau)\rho(\varepsilon_{\theta})\{\hat{\eta}_n(\varepsilon_{\theta},\theta)-\eta_0(\varepsilon_{\theta},\theta)\}\Delta\|
\end{align*}
The first term on the right-hand side converges to zero in outer
probability by the Glivenko-Cantelli property. Further,
\begin{align*}
 \|P I(\varepsilon_{\theta}\leq\tau)\rho(\varepsilon_{\theta})(\hat\eta_n-\eta_0)\Delta \|\leq \|\hat\eta_n-\eta_0\|\|PI(\varepsilon_{\theta}\leq\tau)\rho(\varepsilon_{\theta})\Delta\|=o_{P^*}(1),
\end{align*}
by (\ref{eq:3}). Thus, $\|\Psi_n(\theta,\hat\eta_n)-\Psi(\theta,\eta_0)\|=o_{P^*}(1)$ ,which establishes (\ref{eq:2}) as
\begin{align*}
\|\Psi\{\hat{\theta}_n,\eta_0(\cdot,\hat{\theta}_n)\}\|&\leq \|\Psi\{\hat{\theta}_n,\hat\eta_n(\cdot,\hat{\theta}_n)\}\|+\|\Psi\{\hat{\theta}_n,\hat\eta_n(\cdot,\hat{\theta}_n)\}-\Psi\{\hat{\theta}_n,\eta_0(\cdot,\hat{\theta}_n)\}\|\\
&=o_{P^*}(1)+o_{P^*}(1)=o_{P^*}(1).
\end{align*}
\end{proof}

\begin{lemma}\label{lem:2}
  Let $\hat{\theta}_n$ be an approximate root satisfying
  $\Psi_n\{\hat{\theta}_n,\hat{\eta}_n(\cdot,\hat{\theta}_n)\}=o_{P^*}(n^{-1/2})$. Suppose
  $\Psi\{\theta_0,\eta_0(\cdot,\theta_0)\}$ is differentiable with
  bounded continuous derivative
  $\dot{\Psi}_{\theta}\{\theta_0,\eta_0(\cdot,\theta_0)\}$, and
  $\dot{\Psi}_{\theta}\{\theta_0,\eta_0(\cdot,\theta_0)\}$ is
  non-singular. Then, $\|\hat{\eta}_n-\eta_0\|=O_{P^*}(n^{-1/2})$ and
  $\|\hat{\theta}_n-\theta_0\|=O_{P^*}(n^{-1/2})$.
\end{lemma}

\begin{proof}
First consider the asymptotic representation of $\hat{\eta}_n$,
\begin{align}\label{eq:4}
\nonumber\MoveEqLeft n^{1/2}\{\hat{\eta}_n(t,\theta)-\eta_0(t,\theta)\}\\
\nonumber&=n^{1/2}\left[\frac{1}{d^{(0)}(t,\theta)}\left\{D_{n}^{(1)}(t,\theta)-d^{(1)}(t,\theta)\right\}-\frac{D_{n}^{(1)}(t,\theta)}{D_n^{(0)}(t,\theta)d^{(0)}(t,\theta)}\{D_n^{(0)}(t,\theta)-d^{(0)}(t,\theta)\}\right]\\
\nonumber&=n^{1/2}\left[\frac{1}{d^{(0)}(t,\theta)}\left\{D_{n}^{(1)}(t,\theta)-d^{(1)}(t,\theta)\right\}-\frac{d^{(1)}(t,\theta)}{d^{(0)}(t,\theta)^2}\{D_n^{(0)}(t,\theta)-d^{(0)}(t,\theta)\}\right]+o_{P^*}(1)\\
\nonumber&=d^{(0)}(t,\theta)^{-1}n^{1/2}\left[\left\{D_{n}^{(1)}(t,\theta)-D_n^{(0)}(t,\theta)\eta_0(t,\theta)\right\}-\{d^{(1)}(t,\theta)-d^{(0)}(t,\theta)\eta_0(t,\theta)\}\right]+o_{P^*}(1)\\
&=d^{(0)}(t,\theta)^{-1}\Gn I(\varepsilon_{\theta}\geq t)\left\{(e^{\gamma^{T}Z}X^{T},Z^{T})
-\eta_0(t,\theta)\right\}+o_{P^*}(1).
\end{align}
For the second equality we used that
\begin{align}\label{eq:5}
\|n^{1/2}\{D_{n}^{(k)}(t,\theta)-d^{(k)}(t,\theta)\}\|=O_{P^*}(1),~k=0,1.
\end{align}
This follows from Lemma \ref{lemma:1} which enables the use of
\citet[Theorem 2.14.9]{vdvaartwellner96} from which exponentially
decaying tail bounds are obtained.

The classes of functions
$\{I(\varepsilon_{\theta}\geq
t):t\in]-\infty,\tau],\theta\in\Theta\}$,
$\{\exp(\gamma^T Z)X^{T}:\theta\in\Theta\}$ and $\{Z\}$ are Donsker
and $\eta_0$ is a bounded deterministic function. Thus,
$\{I(\varepsilon_{\theta}\geq t)[\{\exp(\gamma^T
Z)X^{T},Z^{T}\}-\eta_0(t,\theta)]:t\in]-\infty,\tau],\theta\in\Theta\}$
is Donsker. Because $d^{(0)}(t,\theta)^{-1}$ is bounded,
$n^{1/2}\|\hat{\eta}_n-\eta_0\|=O_{P^*}(1)$.

Then,
\begin{align*}
\MoveEqLeft\|n^{1/2}[\Psi_n\{\theta,\hat{\eta}_n(\cdot,\theta)\}-\Psi\{\theta,\eta_0(\cdot,\theta)\}]\|\\
&=\left\|\Gn\psi\{\theta,\hat{\eta}_n(\varepsilon_{\theta},\theta)\}+n^{1/2}P[\psi\{\theta,\hat{\eta}_n(\varepsilon_{\theta},\theta)\}-\psi\{\theta,\eta_0(\varepsilon_{\theta},\theta)\}]\right\|\\
&\leq\|\Gn \psi\{\theta,\hat{\eta}_n(\varepsilon_{\theta},\theta)\}\|+n^{1/2}\|\hat{\eta}_n-\eta_0\|\|PI(\varepsilon_{\theta}\leq\tau)\rho(\epsilon_{\theta})\Delta\|\\
&=O_{P^*}(1)
\end{align*}
Using that $\Psi_n\{\hat{\theta},\hat{\eta}_n(\cdot,\hat{\theta})\}=o_{P^*}(n^{-1/2})$ and $\Psi\{\theta_0,\eta_0(\cdot,\theta_0)\}=0$,
\begin{align*}
O_{P^*}(1)&=-n^{1/2}\left[\Psi_n\{\hat{\theta},\hat{\eta}_n(\cdot,\hat{\theta})\}-\Psi\{\hat{\theta},\eta_0(\cdot,\hat{\theta})\}\right]\\
&=o_{P^*}(1)+n^{1/2}\Psi\{\hat{\theta},\eta_0(\cdot,\hat{\theta})\}-n^{1/2}\Psi\{\theta_0,\eta_0(\cdot,\theta_0)\}\\
&=o_{P^*}(1)+\left[\dot{\Psi}_{\theta}\{\theta_0,\eta_0(\cdot,\theta_0)\}+o_{P^*}(1)\right]n^{1/2}\left(\hat\theta_n-\theta_0\right).
\end{align*}
The invertibilty of $\dot{\Psi}_{\theta}$ gives $n^{1/2}\left(\hat\theta_n-\theta_0\right)=O_{P^*}(1)$.

\end{proof}

\begin{theorem}\label{thm:2}
  Let $\hat{\theta}_n$ be an approximate root satisfying
  $\Psi_n\{\hat{\theta}_n,\hat{\eta}(\cdot,\hat{\theta}_n)\}=o_{P^*}(n^{-1/2})$. Suppose
  that $\theta\mapsto\eta_0(\varepsilon_{\theta},\theta)$ is differentiable with
  uniformly bounded and continuous derivative
  $\dot{\eta}_{\theta}(\varepsilon_{\theta},\theta)$. Then if
  $\dot{\Psi}_{\theta}\{\theta_0,\eta_0(\cdot,\theta_0)\}$ is non-singular,
\begin{align*}
n^{1/2}(\hat{\theta}_{n}-\theta_0)&=-\dot{\Psi}_{\theta}^{-1}\{\theta_{0},\eta_0(\cdot,\theta_0)\}\Gn J(\theta_0,\eta_0,A_0)+o_{P^*}(1).
\end{align*}
\end{theorem}

\begin{proof}
  Lemma \ref{lem:2} shows that there exists a $K<\infty$, such that
  $\|\theta-\theta_0\|\leq Kn^{-1/2}$. Then,
\begin{align}\label{eq:6}
\nonumber\MoveEqLeft n^{1/2}\left[\Psi_n\{\theta,\hat{\eta}_n(\cdot,\theta)\}-\Psi_n\{\theta_0,\hat{\eta}_n(\cdot,\theta_0)\}\right]\\
\nonumber&=n^{1/2}\left[\Pn I(\varepsilon_{\theta}\leq\tau)\rho(\varepsilon_{\theta})\left\{(e^{\gamma^TZ} X^{T},Z^{T})-\hat\eta_n(\varepsilon_{\theta},\theta)\right\}^{T}\Delta\right.\\
\nonumber&\left.\phantom{=+}\;-\Pn I(\varepsilon_{\theta}\leq\tau)\rho(\varepsilon_{\theta})\left\{(e^{\gamma^TZ} X^{T},Z^{T})-\hat\eta_n(\varepsilon_{0},\theta_{0})\right\}^{T}\Delta\right]\\
\nonumber&\phantom{=}+n^{1/2}\left[\Pn I(\varepsilon_{\theta}\leq\tau)\rho(\varepsilon_{\theta})\left\{(e^{\gamma^TZ} X^{T},Z^{T})-\hat\eta_n(\varepsilon_{0},\theta_{0})\right\}^{T}\Delta\right.\\
\nonumber&\left.\phantom{=+}\;-\Pn I(\varepsilon_{\theta}\leq\tau)\rho(\varepsilon_{0})\left\{(e^{\gamma^TZ} X^{T},Z^{T})-\hat\eta_n(\varepsilon_{0},\theta_{0})^{T}\right\}^{T}\Delta\right]\\
\nonumber&\phantom{=}+n^{1/2}\left[\Pn I(\varepsilon_{\theta}\leq\tau)\rho(\varepsilon_{0})\left\{(e^{\gamma^TZ} X^{T},Z^{T})-\hat\eta_n(\varepsilon_{0},\theta_{0})\right\}^{T}\Delta\right.\\
&\left.\phantom{=+}\;-\Pn I(\varepsilon_{0}\leq\tau)\rho(\varepsilon_{0})\left\{(e^{\gamma^TZ} X^{T},Z^{T})-\hat\eta_n(\varepsilon_{0},\theta_{0})\right\}^{T}\Delta\right]\\
\nonumber&\phantom{=}+n^{1/2}\left[\Pn I(\varepsilon_{0}\leq\tau)\rho(\varepsilon_{0})\left\{(e^{\gamma^TZ} X^{T},Z^{T})-\hat\eta_n(\varepsilon_{0},\theta_{0})\right\}^{T}\Delta\right.\\
\nonumber&\left.\phantom{=+}\;-\Pn I(\varepsilon_{0}\leq\tau)\rho(\varepsilon_{0})\left\{(e^{\gamma_{0}^TZ} X^{T},Z^{T})-\hat\eta_n(\varepsilon_{0},\theta_{0})\right\}^{T}\Delta\right].
\end{align}

Consider the first difference on the right-hand side above,
\begin{align}\label{eq:7}
\nonumber\MoveEqLeft -n^{1/2}\Pn I(\varepsilon_{\theta}\leq\tau)\rho(\varepsilon_{\theta})\{\hat{\eta}_n(\varepsilon_{\theta},\theta) - \hat{\eta}_n(\varepsilon_0,\theta_0)\}^{T}\Delta\\
&=-\Gn I(\varepsilon_{\theta}\leq\tau)\rho(\varepsilon_{\theta})\left\{\hat{\eta}_n(\varepsilon_{\theta},\theta)-\hat{\eta}_n(\varepsilon_0,\theta_0)\right\}^{T}\Delta\\
\nonumber&\phantom{=}\;-n^{1/2}P I(\varepsilon_{\theta}\leq\tau)\rho(\varepsilon_{\theta})\left\{\hat{\eta}_n(\varepsilon_{\theta},\theta)-\hat{\eta}_n(\varepsilon_0,\theta_0)\right\}^{T}\Delta.
\end{align}
The class
$\{I(\varepsilon_{\theta}\leq\tau)\rho(\varepsilon_{\theta})\hat{\eta}_n(\varepsilon_{\theta},\theta)^{T}\Delta:\theta\in\Theta\}$
is Donsker by the arguments used in the proof of Theorem \ref{thm:1},
and $I(\varepsilon_{\theta}\leq\tau)\rho(\varepsilon_{\theta})\left\{\hat{\eta}_n(\varepsilon_{\theta},\theta)-\hat{\eta}_n(\varepsilon_{0},\theta_0)\right\}^{T}\Delta$
  converges to zero in $L_2(P)$. Thus, the first term on the
right-hand side above is $o_{P^*}(1)$.

For the second term on the right-hand side of (\ref{eq:7}),
\begin{align}\label{eq:8}
\nonumber\MoveEqLeft n^{1/2}P I(\varepsilon_{\theta}\leq\tau)\rho(\varepsilon_{\theta})\left\{\hat{\eta}_n(\varepsilon_{\theta},\theta)-\hat{\eta}_n(\varepsilon_0,\theta_0)\right\}^{T}\Delta\\
\nonumber&= n^{1/2}P I(\varepsilon_{\theta}\leq\tau)\rho(\varepsilon_{\theta})\left\{\hat{\eta}_n(\varepsilon_{\theta},\theta)-\eta_0(\varepsilon_{\theta},\theta)\right\}^{T}\Delta\\
&\phantom{=}\;- n^{1/2}P I(\varepsilon_{\theta}\leq\tau)\rho(\varepsilon_{\theta})\left\{\hat{\eta}_n(\varepsilon_0,\theta_0)-\eta_0(\varepsilon_0,\theta_0)\right\}^{T}\Delta\\
\nonumber&\phantom{=}\;+n^{1/2}P I(\varepsilon_{\theta}\leq\tau)\rho(\varepsilon_{\theta})\left\{\eta_0(\varepsilon_{\theta},\theta)-\eta_0(\varepsilon_0,\theta_0)\right\}^{T}\Delta.
\end{align}

We now argue that the first two terms on the right-hand side of (\ref{eq:8}) are asymptotically negligible. Similar to (\ref{eq:4}),
\begin{align}\label{eq:10}
\nonumber\MoveEqLeft n^{1/2}P I(\varepsilon_{\theta}\leq\tau)\rho(\varepsilon_{\theta})\left\{\hat{\eta}_n(\varepsilon_{\theta},\theta)-\eta_0(\varepsilon_{\theta},\theta)\right\}^{T}\Delta\\
\nonumber&=n^{1/2}P I(\varepsilon_{\theta}\leq\tau)\rho(\varepsilon_{\theta})\left[\frac{1}{d^{(0)}(\varepsilon_{\theta},\theta)}\left\{D_{n}^{(1)}(\varepsilon_{\theta},\theta_0)-d^{(1)}(\varepsilon_{\theta},\theta)\right\}\right.\\
\nonumber&\phantom{=}\;\left.-\frac{D_n^{(1)}(\varepsilon_{\theta},\theta)}{D_n^{(0)}(\varepsilon_{\theta},\theta)d^{(0)}(\varepsilon_{\theta},\theta)}\left\{D_{n}^{(0)}(\varepsilon_{\theta},\theta)-d^{(0)}(\varepsilon_{\theta},\theta)\right\}\right]^{T}\Delta\\
&=n^{1/2}\int I(t'\leq\tau)\rho(t')\left[\frac{1}{d^{(0)}(t',\theta)}\left\{D_{n}^{(1)}(t',\theta_0)-d^{(1)}(t',\theta)\right\}\right.\\
\nonumber&\phantom{=}\;\left.-\frac{D_n^{(1)}(t',\theta)}{d^{(0)}(t',\theta)d^{(0)}(t',\theta)}\left\{D_{n}^{(0)}(t',\theta)-d^{(0)}(t',\theta)\right\}\right]^{T}\delta dP_{\varepsilon_0,\Delta,X,Z}(t,\delta,x,z)
\end{align}
where $t'=t'(\theta,x,z)=\exp\{(\gamma-\gamma_0)^Tz\}t+\exp\{\gamma^Tz\}(\beta-\beta_0)^Tx$ and $P_{\varepsilon_0,\Delta,X,Z}$ is the joint probability law of $(\varepsilon_0,\Delta,X,Z)$.

Now, as in \citet{nan09}[p. 2368],
\begin{align*}
\MoveEqLeft \left\|n^{1/2}\int I(t'\leq\tau)\rho(t')\left[\frac{1}{d^{(0)}(t',\theta)}\left\{D_{n}^{(1)}(t',\theta)-d^{(1)}(t',\theta)\right\}\right.\right.&\\
&\phantom{=}\;\left.-\frac{D_n^{(1)}(t',\theta)}{D_n^{(0)}(t',\theta)d^{(0)}(t',\theta)}\left\{D_{n}^{(0)}(t',\theta)-d^{(0)}(t',\theta)\right\}\right]^{T}\delta dP_{\varepsilon_0,\Delta,X,Z}(t,\delta,x,z)\\
&-n^{1/2}\int I(t'\leq\tau)\rho(t')\left[\frac{1}{d^{(0)}(t',\theta)}\left\{D_{n}^{(1)}(t',\theta)-d^{(1)}(t',\theta)\right\}\right.\\
&\phantom{=}\;\left.\left.-\frac{d^{(1)}(t',\theta)}{d^{(0)}(t',\theta)^2}\left\{D_{n}^{(0)}(t',\theta)-d^{(0)}(t',\theta)\right\}\right]^{T}\delta dP_{\varepsilon_0,\Delta,X,Z}(t,\delta,x,z)\right\|\\
&=\left\|n^{1/2}\int I(t'\leq\tau)\rho(t')\left\{\frac{d^{(1)}(t',\theta)}{d^{(0)}(t',\theta)^2}-\frac{D_n^{(1)}(t',\theta)}{D_n^{(0)}(t',\theta)d^{(0)}(t',\theta)}\right\}^{T}\right.\\
&\phantom{=}\;\left.\times\left\{D_{n}^{(0)}(t',\theta)-d^{(0)}(t',\theta)\right\}\delta dP_{\varepsilon_0,\Delta,X,Z}(t,\delta,x,z)\right\|\\
&\leq\left\|\rho(t)\right\|\left\|\frac{d^{(1)}(t,\theta)}{d^{(0)}(t,\theta)^2}-\frac{D_n^{(1)}(t,\theta)}{D_n^{(0)}(t,\theta)d^{(0)}(t,\theta)}\right\|\left\|n^{1/2}\left\{D_{n}^{(0)}(t,\theta)-d^{(0)}(t,\theta)\right\}\right\|\\
&=O_{P^*}(1)o_{P^*}(1)O_{P^*}(1)=o_{P^*}(1).
\end{align*}
Thus, similar to (\ref{eq:4}), (\ref{eq:10}) is
\begin{align*}
  \MoveEqLeft \int I(t'\leq\tau)\rho(t')\delta d^{(0)}(t',\theta)^{-1}\\
&\phantom{=}\;\times\Gn I(\varepsilon_{\theta}\geq t')\left[\left\{\exp(\gamma^T Z)X^T, Z^T\right\}-\eta_0(t',\theta)\right]^{T}
  dP_{\varepsilon_0,\Delta,X,Z}(t,\delta,x,z)\\
&=\int\Gn I(t'\leq\tau)\rho(t')\ell(t',\theta,X,Z,\varepsilon_{\theta})dP_{\varepsilon_0,\Delta,X,Z}(t,1,x,z)
\end{align*}
where $\ell(t',\theta,X,Z,\varepsilon_{\theta})=d^{(0)}(t',\theta)^{-1}I(\varepsilon_{\theta}\geq t')\left[\left\{\exp(\gamma^T Z)X^T, Z^T\right\}-\eta_0(t',\theta)\right]^{T}$.

Similarly,
\begin{align*}
\MoveEqLeft n^{1/2}P I(\varepsilon_{\theta}\leq\tau )\rho(\varepsilon_{\theta})\left\{\hat{\eta}_n(\varepsilon_0,\theta_0)-\eta_0(\varepsilon_0,\theta_0)\right\}^{T}\Delta\\
&=\int\Gn I(t'\leq\tau )\rho(t')\ell(t,\theta_0,X,Z,\varepsilon_{0})dP_{\varepsilon_0,\Delta,X,Z}(t,1,x,z)+o_{P^*}(1).
\end{align*}
Thus, the first two terms on the right-hand side of (\ref{eq:8}) equate to
\begin{align*}
\int\Gn I(t'\leq\tau)\rho(t')\left\{\ell(t',\theta,X,Z,\varepsilon_{\theta})-\ell(t,\theta_0,X,Z,\varepsilon_{0})\right\}dP_{\varepsilon_0,\Delta,X,Z}(t,1,x,z)
\end{align*}
which is $o_{P^*}(1)$ as
$\{\ell(t,\theta,X,Z,\varepsilon_{\theta}):t\in ]-\infty,\tau],\theta\in\Theta\}$ is a
Donsker class, and
$\ell(t',\theta,X,Z,\varepsilon_{\theta})-\ell(t,\theta_0,X,Z,\varepsilon_0)$
converges to zero in $L_2(P)$.

Thus, (\ref{eq:8}) is
\begin{align}\label{eq:11}
\nonumber \MoveEqLeft n^{1/2}P I(\varepsilon_{\theta}\leq\tau )\rho(\varepsilon_{\theta})\left\{\eta_0(\varepsilon_{\theta},\theta)-\eta_0(\varepsilon_0,\theta_0)\right\}^{T}\Delta+o_{P^*}(1)\\
&=n^{1/2}\left\{P I(\varepsilon_{\theta}\leq\tau )\rho(\varepsilon_{\theta})\dot{\eta}_{\theta}(\varepsilon_0,\theta_0)^{T}+o_{P^*}(1)\right\}\Delta(\theta-\theta_0) +o_{P^*}(1)\\
\nonumber &=\left\{P I(\varepsilon_{0}\leq\tau )\rho(\varepsilon_{0})\dot{\eta}_{\theta}(\varepsilon_0,\theta_0)^{T}\Delta\right\} n^{1/2}(\theta-\theta_0) +o_{P^*}(1).
\end{align}
The first equality in (\ref{eq:11}) follows from using the assumption
of bounded density functions for failure and censoring times, that is,
assumptions \ref{ass:3} and \ref{ass:4}, together with the dominated
convergence theorem.

The second difference on the right-hand side of (\ref{eq:6}) is
\begin{align}\label{eq:12}
\nonumber \MoveEqLeft n^{1/2}\Pn I(\varepsilon_{\theta}\leq\tau)\left\{\rho(\varepsilon_{\theta})-\rho(\varepsilon_{0})\right\}\left\{(e^{\gamma^TZ} X^{T},Z^{T})-\hat\eta_n(\varepsilon_{0},\theta_{0})\right\}^{T}\Delta\\\
&=\Pn I(\varepsilon_{\theta}\leq\tau)diag\left\{(e^{\gamma^TZ} X^{T},Z^{T})-\hat\eta_n(\varepsilon_{0},\theta_{0})\right\}\left\{\dot{\rho}_{\theta}(\varepsilon_{0})+o_{P^*}(1)\right\}\Delta n^{1/2}(\theta-\theta_0)\\
\nonumber &=\left[PI(\varepsilon_{0}\leq\tau)diag\left\{(e^{\gamma_{0}^T Z} X^{T},Z^{T})-\eta_{0}(\varepsilon_{0},\theta_{0})\right\}\dot{\rho}_{\theta}(\varepsilon_{0}) \Delta\right] n^{1/2}(\theta-\theta_0)+o_{P^*}(1).
\end{align}

For the third difference on the right-hand side of (\ref{eq:6}) let
$h(t,\theta,X,Z)=pr(\varepsilon_{\theta}\leq t,\:\Delta=1|X,Z)$. Then,
according to assumptions 2, 3 and 4, $h$ has bounded continuous
derivative $\dot{h}_{\theta}(t,\theta,X,Z)$ with respect to $\theta$. We will establish
\begin{align}\label{eq:13}
\MoveEqLeft n^{1/2}\Pn \left\{I(\varepsilon_{\theta}\leq\tau)-I(\varepsilon_{0}\leq\tau)\right\}\rho(\varepsilon_{0})\left\{(e^{\gamma^TZ} X^{T},Z^{T})-\hat\eta_n(\varepsilon_{0},\theta_{0})\right\}^{T}\Delta\\\
\nonumber&=\left[P\rho(\tau)\left\{(e^{\gamma_{0}^T Z} X^{T},Z^{T})-\eta_{0}(\tau,\theta_{0})^{T}\right\}\dot{h}_{\theta}(\tau,\theta_{0},X,Z)\right] n^{1/2}(\theta-\theta_0)+o_{P^*}(1).
\end{align}
To see this, first note that
\begin{align*}
\MoveEqLeft n^{1/2}\Pn \left\{I(\varepsilon_{\theta}\leq\tau)-I(\varepsilon_{0}\leq\tau)\right\}\rho(\varepsilon_{0})\left\{(e^{\gamma^TZ} X^{T},Z^{T})-\hat\eta_n(\varepsilon_{0},\theta_{0})\right\}^{T}\Delta\\
&=\Gn \left\{I(\varepsilon_{\theta}\leq\tau)-I(\varepsilon_{0}\leq\tau)\right\}\rho(\varepsilon_{0})\left\{(e^{\gamma^TZ} X^{T},Z^{T})-\hat\eta_n(\varepsilon_{0},\theta_{0})\right\}^{T}\Delta\\
&\phantom{=}\;+n^{1/2}P \left\{I(\varepsilon_{\theta}\leq\tau)-I(\varepsilon_{0}\leq\tau)\right\}\rho(\varepsilon_{0})\left\{(e^{\gamma^TZ} X^{T},Z^{T})-\hat\eta_n(\varepsilon_{0},\theta_{0})\right\}^{T}\Delta.
\end{align*}
By similar arguments as above the first term on the right rand side is $o_{P^*}(1) $. For the second term on the right hand side we have
\begin{align*}
\MoveEqLeft n^{1/2}P \left\{I(\varepsilon_{\theta}\leq\tau)-I(\varepsilon_{0}\leq\tau)\right\}\rho(\varepsilon_{0})\left\{(e^{\gamma^TZ} X^{T},Z^{T})-\hat\eta_n(\varepsilon_{0},\theta_{0})\right\}^{T}\Delta\\
&=n^{1/2}P \left\{I(\varepsilon_{\theta}\leq\tau)-I(\varepsilon_{0}\leq\tau)\right\}\rho(\varepsilon_{0})\left\{(e^{\gamma_{0}^TZ} X^{T},Z^{T})-\eta_{0}(\varepsilon_{0},\theta_{0})\right\}^{T}\Delta\\
&\phantom{=}\;+n^{1/2}P \left\{I(\varepsilon_{\theta}\leq\tau)-I(\varepsilon_{0}\leq\tau)\right\}\rho(\varepsilon_{0})\left\{(e^{\gamma^TZ} X^{T},Z^{T})-(e^{\gamma_{0}^TZ} X^{T},Z^{T})\right\}^{T}\Delta\\
&\phantom{=}\;-n^{1/2}P \left\{I(\varepsilon_{\theta}\leq\tau)-I(\varepsilon_{0}\leq\tau)\right\}\rho(\varepsilon_{0})\left\{\hat\eta_n(\varepsilon_{0},\theta_{0})-\eta_{0}(\varepsilon_{0},\theta_{0})\right\}^{T}\Delta.
\end{align*}
The second term on the right hand side is clearly $o_{P^*}(1) $. For the third term note that, similar to (\ref{eq:10}), we have
\begin{align*}
\MoveEqLeft n^{1/2}P \left\{I(\varepsilon_{\theta}\leq\tau)-I(\varepsilon_{0}\leq\tau)\right\}\rho(\varepsilon_{0})\left\{\hat\eta_n(\varepsilon_{0},\theta_{0})-\eta_{0}(\varepsilon_{0},\theta_{0})\right\}^{T}\\
&=\int\Gn \{I(t'\leq\tau)-I(t\leq\tau)\}\rho(t)\ell(t,\theta_0,X,Z,\varepsilon_0)dP_{\varepsilon_0,\Delta,X,Z}(t,1,x,z)+o_{P^*}(1).
\end{align*}
Accordingly this term is also $o_{P^*}(1) $ since $\{I(t'\leq\tau)-I(t\leq\tau)\}\rho(t)\ell(t,\theta_0,X,Z,\varepsilon_0)$ converges to zero in $L_2(P)$.
Finally, for the first term we have
\begin{align*}
\MoveEqLeft n^{1/2}P \left\{I(\varepsilon_{\theta}\leq\tau)-I(\varepsilon_{0}\leq\tau)\right\}\rho(\varepsilon_{0})\left\{(e^{\gamma_{0}^TZ} X^{T},Z^{T})-\eta_{0}(\varepsilon_{0},\theta_{0})\right\}^{T}\Delta\\
&=\left[P\rho(\tau)\left\{(e^{\gamma_{0}^T Z} X^{T},Z^{T})-\eta_{0}(\tau,\theta_{0})^{T}\right\}\dot{h}_{\theta}(\tau,\theta_{0},X,Z)\right] n^{1/2}(\theta-\theta_0)+o_{P^*}(1),
\end{align*}
where we have used that for a continuous density $p$ on the real line, continuous function $f$ with $\int_{-\infty}^{\infty}|f(t)|p(t)dt<\infty$, and continuously differentiable function $g(\theta)$ we have:
\begin{align*}
\int_{-\infty}^{g(\theta)}f(s)p(s)ds-\int_{-\infty}^{g(\theta_{0})}f(s)p(s)ds &=f\{g(\theta_{0})\}p\{g(\theta_{0})\}\{g(\theta)-g(\theta_{0})\}+o(|g(\theta)-g(\theta_{0}|)\\
&=f\{g(\theta_{0})\}p\{g(\theta_{0})\}\dot{g}_{\theta}(\theta_{0})(\theta-\theta_{0})+o(\|\theta-\theta_{0}\|).
\end{align*}
Now, consider the last difference on the right-hand side of (\ref{eq:6}):
\begin{align}\label{eq:14}
\nonumber\MoveEqLeft n^{1/2}\Pn I(\varepsilon_{0}\leq \tau)\rho(\varepsilon_0)\left[\begin{array}{c}
\{\exp(\gamma^TZ)-\exp(\gamma_0^TZ)\}X\\
0_{q\times 1}
\end{array}\right]\Delta\\
\nonumber&=\Pn I(\varepsilon_{0}\leq \tau)\rho(\varepsilon_0)\left\{\begin{array}{cc}
0_{p\times q} & XZ^T\exp(\gamma_0^TZ)\\
0_{q\times q} & 0_{q\times p}
\end{array}\right\}\Delta n^{1/2}(\theta-\theta_0)+o_{P^*}(1)\\
&=\left[P I(\varepsilon_{0}\leq \tau)\rho(\varepsilon_0)\left\{\begin{array}{cc}
0_{p\times q} & XZ^T\exp(\gamma_0^TZ)\\
0_{q\times q} & 0_{q\times p}
\end{array}\right\}\Delta \right]n^{1/2}(\theta-\theta_0)+o_{P^*}(1).
\end{align}
From (\ref{eq:11}), (\ref{eq:12}), (\ref{eq:13}), and (\ref{eq:14}), (\ref{eq:6}) is
\begin{align*}
n^{1/2}\left[\Psi_n\{\theta,\hat{\eta}_n(\cdot,\theta)\}-\Psi_n\{\theta_0,\hat{\eta}_n(\cdot,\theta_0)\}\right]=\dot{\Psi}_{\theta}\{\theta_0,\eta_0(\cdot,\theta_0)\} n^{1/2}(\theta-\theta_0)+o_{P^*}(1).
\end{align*}
On the other hand, inserting $\hat{\theta}_n$ into the left-hand side above,
\begin{align*}
\MoveEqLeft n^{1/2}\Psi_n\{\hat{\theta}_n,\hat{\eta}_n(\cdot,\hat{\theta}_n)\}-n^{1/2}\Psi_n\{\theta_0,\hat{\eta}_n(\cdot,\theta_0)\}\\
&=o_{P^*}(1)-n^{1/2}\Psi_n\{\theta_0,\hat{\eta}_n(\cdot,\theta_0)\}\\
&=o_{P^*}(1)-\Gn \psi\{O;\theta_0,\eta_0(\cdot,\theta_0)\}+\Gn I(\varepsilon_{0}\leq\tau) \rho(\varepsilon_{0})\left\{\hat{\eta}_n(\varepsilon_0,\theta_0) - \eta_0(\varepsilon_0,\theta_0)\right\}^{T}\Delta\\
&\phantom{=}\;+n^{1/2}P  I(\varepsilon_{0}\leq\tau) \rho(\varepsilon_{0})\left\{\hat{\eta}_n(\varepsilon_0,\theta_0) - \eta_0(\varepsilon_0,\theta_0)\right\}^{T}\Delta.
\end{align*}
The third term on the right-hand side above is $o_{P^*}(1)$. For the last term above,
\begin{align*}
\MoveEqLeft n^{1/2}P I(\varepsilon_{0}\leq\tau) \rho(\varepsilon_{0})\left\{\hat{\eta}_n(\varepsilon_0,\theta_0) - \eta_0(\varepsilon_0,\theta_0)\right\}^{T}\Delta\\
&=n^{1/2}\int_{-\infty}^{\tau} \rho(t)d^{(0)}(t,\theta)^{-1}\left[\left\{D_{n}^{(1)}(t,\theta_0)-D_n^{(0)}(t,\theta_0)\eta_0(t,\theta_0)\right\}\right.\\
&\phantom{=}\;\left.-\left\{d^{(1)}(t,\theta_0)-d^{(0)}(t,\theta_0)\eta_0(t,\theta_0)\right\}\right]^{T} dP_{\varepsilon_0,\Delta}(t,1)+o_{P^*}(1)\\
&=\int_{-\infty}^{\tau} \rho(t)d^{(0)}(t,\theta_0)^{-1}\Gn I(\varepsilon_{0}\geq t)\left[\left\{\exp(\gamma_{0}^{T}Z)X^{T},Z^{T}\right\}-\eta_{0}(t,\theta_{0})\right]^{T} dP_{\varepsilon_0,\Delta}(t,1)+o_{P^*}(1)\\
&=\Gn\int_{-\infty}^{\tau} \rho(t) I(\varepsilon_{0}\geq t)\left[\left\{\exp(\gamma_{0}^{T}Z)X^{T},Z^{T}\right\}-\eta_{0}(t,\theta_{0})\right]^{T} dA_0(t)+o_{P^*}(1),
\end{align*}
where $A_0$ is the cumulative hazard of the error term
$e=e^{\gamma_0^TZ}(\log T+\beta_0^T X)$.

Thus, combining the three displays above
\begin{align*}
\MoveEqLeft n^{1/2}(\hat{\theta}_{n}-\theta_0)\\
&=-\dot{\Psi}_{\theta}^{-1}\{\theta_0,\eta_0(\cdot,\theta_0)\}\\
&\phantom{=}\;\times\Gn\left(\psi\{O;\theta_0,\eta_0(\cdot,\theta_0)\}-\int_{-\infty}^{\tau} \rho(t)I(\varepsilon_{0}\geq t)\{(\exp(\gamma_{0}^{T}Z)X^{T},Z^{T})-\eta(t,\theta_{0})\}^{T}dA(t)\right)\\
&\phantom{=}\;+o_{P^*}(1).
\end{align*}
\end{proof}

\begin{theorem}
  Let $\hat{\theta}_n$ be an estimator of $\theta_{0}$ such that
  $n^{1/2}(\hat{\theta}_n-\theta_{0})$ converges weakly to a zero mean
  normal distribution. Then
  $n^{1/2}\{\hat{A}_{n}(\hat{\theta}_{n},\cdot)-A_0(\cdot)\}$ converges
  weakly to a tight zero mean Gaussian process on $]-\infty,\tau]$ and
  the following holds
\begin{align*}
  n^{1/2}\{\hat{A}_{n}(\hat{\theta}_{n},t)-A_0(t)\}&=\dot{\phi}_{\theta}(t,\theta_{0})n^{1/2}(\hat{\theta}_{n}-\theta_{0})+\Gn H(t,\theta_0,d^{(0)},A_0)+o_{P^*}(1).
\end{align*}
\end{theorem}

\begin{proof}
Let $\|\theta-\theta_{0}\|<K n^{-1/2}$ for some $K<\infty$. Then note that

\begin{equation*}
  n^{1/2}\{\hat{A}(t,\theta)-\hat{A}(t,\theta_{0})\}=n^{1/2}\Pn \left\{\frac{I(\varepsilon_{\theta}\leq t)\Delta}{D_n^{(0)}(\varepsilon_{\theta},\theta)}-\frac{I(\varepsilon_{0}\leq t)\Delta}{D_n^{(0)}(\varepsilon_{0},\theta_{0})}\right\}.
\end{equation*}
As in (\ref{eq:8}) in the proof of Theorem \ref{thm:2} one may show that
\begin{align*}
\MoveEqLeft n^{1/2}\Pn \left[ I(\varepsilon_{\theta}\leq
  t)\Delta\left\{\frac{1}{D_n^{(0)}(\varepsilon_{\theta},\theta)}-\frac{1}{D_n^{(0)}(\varepsilon_{0},\theta_{0})}\right\}\right]\\
&=-n^{1/2}P\left[ I(\varepsilon_{0}\leq
  t)\Delta\frac{\dot{d}_{\theta}^{(0)}(\varepsilon_{0},\theta_{0})}{d^{(0)}(\varepsilon_{0},\theta_{0})^{2}}\right](\theta-\theta_{0})+o_{P^*}(1).
\end{align*}
As in (\ref{eq:13}) in the proof of Theorem \ref{thm:2} one may also show that
\begin{align*}
n^{1/2}\Pn \left[\frac{\Delta\{I(\varepsilon_{\theta}\leq t)-I(\varepsilon_{0}\leq t)\}}{D_n^{(0)}(\varepsilon_{0},\theta_{0})}\right]=n^{1/2}P\left\{\frac{\dot{h}_{\theta}(t,\theta_{0},X,Z)}{d^{(0)}(t,\theta_{0})}\right\}(\theta-\theta_{0})+o_{P^*}(1).
\end{align*}
Combining the displays above, we get
\begin{equation}\label{eq:15}
n^{1/2}\{\hat{A}(t,\hat{\theta}_{n})-\hat{A}(t,\theta_{0})\}=\dot{\phi}_{\theta}(t,\theta_{0})n^{1/2}(\hat{\theta}_{n}-\theta_{0})+o_{P^*}(1).
\end{equation}

Secondly note that
\begin{align}\label{eq:16}
\nonumber \MoveEqLeft n^{1/2}\{\hat{A}(t,\theta_{0})-A_0(t)\}\\
\nonumber&=\Gn \left[ I(\varepsilon_{0}\leq t)\Delta \left\{\frac{1}{D_n^{(0)}(\varepsilon_{0},\theta_{0})}-\frac{1}{d^{(0)}(\varepsilon_{0},\theta_{0})}\right\}\right]\\
&\phantom{=}\;+n^{1/2}P \left[ I(\varepsilon_{0}\leq t)\Delta \left\{\frac{1}{D_n^{(0)}(\varepsilon_{0},\theta_{0})}-\frac{1}{d^{(0)}(\varepsilon_{0},\theta_{0})}\right\}\right]+\Gn \left\{\frac{I(\varepsilon_{0}\leq t)\Delta}{d^{(0)}(\varepsilon_{0},\theta_{0})}\right\}.
\end{align}
By similar arguments as in the proof of Theorem \ref{thm:2} one may show that the first term on the right hand side of (\ref{eq:16}) converges in probability to $0$.  For the second term note that

\begin{align*}
\MoveEqLeft n^{1/2}P \left[ I(\varepsilon_{0}\leq t)\Delta \left\{\frac{1}{D_n^{(0)}(\varepsilon_{0},\theta_{0})}-\frac{1}{d^{(0)}(\varepsilon_{0},\theta_{0})}\right\}\right]\\
&=n^{1/2}\int_{-\infty}^{t} \left\{\frac{1}{D_n^{(0)}(s,\theta_{0})}-\frac{1}{d^{(0)}(s,\theta_{0})}\right\}dP_{\varepsilon_{0},\Delta}(s,1)\\
&=-n^{1/2}\int_{-\infty}^{t}\frac{1}{D_n^{(0)}(s,\theta_{0})d^{(0)}(s,\theta_{0})}\left\{D^{(0)}(s,\theta_{0})-d^{(0)}(s,\theta_{0})\right\}dP_{\varepsilon_{0},\Delta}(s,1)\\
&=-n^{1/2}\int_{-\infty}^{t}\frac{1}{D_n^{(0)}(s,\theta_{0})}\left\{D_n^{(0)}(s,\theta_{0})-d^{(0)}(s,\theta_{0})\right\}dA(s)\\
&=-\int_{-\infty}^{t}\frac{1}{d^{(0)}(s,\theta_{0})}\Gn I(\varepsilon_{0}\geq s)dA_0(s)+o_{P^*}(1)\\
&=-\Gn\int_{-\infty}^{t}\frac{1}{d^{(0)}(s,\theta_{0})}I(\varepsilon_{0}\geq s)dA_0(s)+o_{P^*}(1),
\end{align*}
from which it follows by Lemma \ref{lemma:1} that the second term
converges weakly to a tight zero mean Gaussian process. It also
follows from Lemma \ref{lemma:1} that the third term in (\ref{eq:16}) converges
weakly to a tight zero mean Gaussian process.  Combining (\ref{eq:15})
and (\ref{eq:16}) we have that
$n^{1/2}\{\hat{A}(\hat{\theta}_{n},\cdot)-A(\cdot)\}$ converges weakly
to a tight zero mean Gaussian process on $]-\infty,\tau]$ and that
\begin{align*}
\MoveEqLeft n^{1/2}\{\hat{A}_{n}(\hat{\theta}_{n},t)-A_0(t)\}=\\
&\dot{\phi}_{\theta}(t,\theta_{0})n^{1/2}(\hat{\theta}_{n}-\theta_{0})+\Gn\left\{\frac{I(\varepsilon_{0}\leq t)}{\Delta d^{(0)}(\varepsilon_{0},\theta_{0})} -\int_{-\infty}^{t}\frac{I(\varepsilon_{0}\geq s)}{d^{(0)}(s,\theta)}dA_0(s) \right\}+o_{P^*}(1).
\end{align*}

\end{proof}

\section{Conditional multiplier method}

While it is straightforward to compute confidence intervals for
$\theta_0$ and $A_0(t)$ (once we can estimate the limiting variances as
discussed in the main paper), we now discuss how confidence intervals
for functions of $\theta_0$ and $A_0(t)$ can be produced. For this, we
rely on empirical process theory via the conditional multiplier method
(cf.~\citet{vdvaartwellner96}).

From empirical process theory we have that for a Gaussian vector
$G=(G_1,\ldots,G_n)$, the limiting distribution of
\begin{align*}
-\dot{\Psi}_{\theta}^{-1} n^{-1/2} \sum_{i=1}^n J(O_i) G_i
\end{align*}
is the same as that of $n^{1/2}(\hat{\theta}_{n}-\theta_0)$. Defining
$G^b$ to be one such randomly generated Gaussian vector
($b = 1,\ldots,m$), we may compute
\begin{align}
\hat\theta^b = \hat\theta_n-\hat{\dot{\Psi}}_{\theta}^{-1} n^{-1} \sum_{i=1}^n \hat J_i G_i^b. \label{parsim}
\end{align}
Hence, generating a sample of $G^b$ vectors produces a sample of
$\hat \theta^b$ vectors. Note that the quantiles of the sample
$\{\hat \theta_j^1, \ldots, \hat \theta_j^m\}$ may be used to form
confidence intervals for $\theta_j$.

We now turn to $A_0(t)$ where we have that
\begin{align*}
n^{-1/2} \sum_{i=1}^n \left[-\dot{\phi}_{\theta}(t) \dot{\Psi}_{\theta}^{-1} J(O_i) + H(O_i;t) \right] G_i
\end{align*}
has the same limiting distribution as
$n^{1/2}\{\hat{A}_{n}(\hat{\theta}_{n},t)-A_0(t)\}$. However, it is
well known that transforming to the unrestricted $\log A_0(t)$ scale
(and subsequently back-transforming) is preferable. Thus, we consider
\begin{align*}
n^{-1/2} \frac{1}{ \hat A_n(\hat\theta_n, t) } \sum_{i=1}^n \left[-{\dot{\phi}_{\theta}}(t) {\dot{\Psi}}_{\theta}^{-1} J(O_i) + H(O_i;t) \right] G_i
\end{align*}
which has the same limiting distribution as
$n^{1/2}\{\log \hat{A}_{n}(\hat{\theta}_{n},t)-\log A_0(t)\}$. Hence,
computing
\begin{align}
  \hat{A}^b(t) = \hat{A}_n(\hat{\theta}_n,t)\exp\left\{n^{-1} \frac{1}{\hat A_n(\hat\theta_n, t) } \sum_{i=1}^n \left[-\hat{\dot{\phi}}_{\theta}(t) \hat{\dot{\Psi}}_{\theta}^{-1} \hat J_i + \hat H_i(t) \right] G_i^b\right\} \label{Asim}
\end{align}
for $b=1,\ldots,m$ creates a sample whose quantiles may be used to produce confidence bands for $A_0(t)$.


When applying (\ref{parsim}) and (\ref{Asim}) above, we must maintain
the same set of Gaussian vectors, $\{G^1,\ldots,G^m\}$, i.e., both
$\hat\theta^b$ and $\hat{A}^b(t)$ are generated from the same Gaussian
vector $G^b$ (for $b=1,\ldots,m$). This has the effect of respecting
the dependence structure between the estimators $\hat \theta_n$ and
$\hat{A}_n(\hat{\theta}_n,t)$ which propagates into any functions of these
estimates. Because the limiting distribution of
$\{\hat\theta^b, \hat A^b(t)\}$ is the same as that of
$\{\hat\theta_n, \hat A_n(\hat\theta_n, t)\}$, we have, from the
continuous mapping theorem, that the limiting distribution of
$w\{\hat\theta^b, \hat A^b(t)\}$ is the same as that of
$w\{\hat\theta_n, \hat A_n(\hat\theta_n, t)\}$ where $w(\cdot,\cdot)$
is a continuous function of the parameters and error cumulative hazard
function. Hence, from the simulated sample
$\{\hat\theta^b, \hat A^b(t)\}$, $b=1,\ldots,m$, we may produce
confidence bands for any functional of interest. As an example,
consider the conditional survivor function for our proposed model
which is given by
\begin{align*}
S(t\,|\,x_i,z_i) = \exp\left\{-A_0\left( \frac{\log t - \mu_{i0}}{\sigma_{i0}} \right)\right\}
\end{align*}
where $\mu_{i0} = - x_i^T \beta_0$ and $\sigma_{i0} = \exp(-z_i^T\gamma_0)$.
Hence, we can compute
\begin{align*}
\hat S^b(t\,|\,x_i,z_i) = \exp\left\{-\hat A^b\left( \frac{\log t - \hat\mu^b_i}{\hat\sigma^b_i} \right)\right\}
\end{align*}
where $\hat\mu^b_i = - x_i^T \hat\beta^b$ and
$\hat\sigma^b_i = \exp(-z_i^T\hat\gamma^b)$, $b=1,\ldots,m$, from
which confidence bands can be produced.

\section{Additional simulation results}

In Section 4.1 of the main paper, we presented a subset of a larger simulation study, the results of which are contained here. The details of the full simulation study are as described in the main paper with the addition of the sample sizes $n=50$ and $n=500$, and, furthermore, the Gehan weight, $\rho(u) = \sum_{j=1}^nY_j^*(u)/n$, was also considered.


Tables \ref{tab:reslr} - \ref{tab:reseff} below display the bias and coverage percentages for each of the three weight function types, while Tables \ref{tab:lreff} - \ref{tab:effeff} show the empirical and estimated standard errors. In all cases the bias is low, the coverage is close to the nominal level, and our proposed variance estimators are adequately capturing the true variations in estimation (and, indeed, the efficiency is similar across the three weight function choices).

\begin{table}[!htbp]
\begin{center}
\caption{Log-rank bias and coverage\label{tab:reslr}}
\centering
\begin{tabular}{cccrrrrrr}
\hline
&&& \multicolumn{2}{c}{$n=50$} & \multicolumn{2}{c}{$n=100$} & \multicolumn{2}{c}{$n=500$} \\
$\tau$ & Cens. & Parameter & Bias & Cov. & Bias & Cov. & Bias & Cov. \\
\hline
2        & 20\% & $\beta_1$  &  0.004 & 94.6 &  0.000 & 94.8 &  0.000 & 95.0 \\
               && $\beta_2$  &  0.017 & 94.3 &  0.008 & 94.0 & -0.002 & 94.8 \\
               && $\gamma_1$ & -0.033 & 93.8 & -0.001 & 94.1 &  0.000 & 94.6 \\
               && $A$        & -0.016 & 94.9 & -0.009 & 94.5 & -0.002 & 95.1 \\
               && $S$        &  0.000 & 96.6 &  0.000 & 94.6 &  0.000 & 94.8 \\
               && $r$        &  0.007 & 94.8 &  0.003 & 95.1 &  0.001 & 95.0 \\
2        & 50\% & $\beta_1$  & -0.004 & 95.6 &  0.000 & 94.2 &  0.000 & 94.8 \\
               && $\beta_2$  &  0.029 & 95.2 &  0.005 & 93.3 & -0.001 & 94.1 \\
               && $\gamma_1$ & -0.064 & 93.1 & -0.015 & 94.6 & -0.004 & 94.3 \\
               && $A$        & -0.018 & 93.5 & -0.010 & 95.3 & -0.003 & 94.5 \\
               && $S$        &  0.004 & 97.1 &  0.008 & 95.3 &  0.001 & 94.5 \\
               && $r$        &  0.016 & 95.0 &  0.006 & 95.5 &  0.000 & 94.8 \\
$\infty$ & 20\% & $\beta_1$  &  0.000 & 95.1 &  0.000 & 94.3 & -0.001 & 95.4 \\
               && $\beta_2$  &  0.017 & 94.8 &  0.003 & 93.6 &  0.000 & 94.6 \\
               && $\gamma_1$ & -0.036 & 93.9 & -0.007 & 93.4 &  0.002 & 94.3 \\
               && $A$        & -0.012 & 94.7 & -0.004 & 94.3 & -0.001 & 94.9 \\
               && $S$        & -0.002 & 96.2 &  0.001 & 94.5 &  0.000 & 95.2 \\
               && $r$        &  0.006 & 95.1 &  0.004 & 95.3 &  0.001 & 95.2 \\
$\infty$ & 50\% & $\beta_1$  & -0.005 & 95.0 &  0.003 & 94.5 & -0.001 & 94.9 \\
               && $\beta_2$  &  0.031 & 95.2 &  0.001 & 93.1 &  0.000 & 94.5 \\
               && $\gamma_1$ & -0.066 & 92.8 & -0.016 & 93.7 &  0.000 & 94.4 \\
               && $A$        & -0.016 & 94.5 & -0.010 & 95.1 & -0.002 & 94.8 \\
               && $S$        &  0.002 & 97.3 &  0.009 & 95.3 &  0.002 & 94.6 \\
               && $r$        &  0.018 & 94.9 &  0.005 & 95.4 &  0.001 & 94.7 \\
\hline
\end{tabular}

{\footnotesize
Cens., censored proportion; Bias, median bias; Cov., empirical coverage percentage for 95\% confidence interval; $A = A(0)$; $S=S(t^{(1)}_{0.5}\,|\, x^{(1)})$; $r = r(x^{(1)},x^{(2)})$.}
\end{center}
\end{table}

\begin{table}[p]
\begin{center}
\caption{Gehan bias and coverage\label{tab:resg}}
\begin{tabular}{cccrrrrrr}
\hline
&&& \multicolumn{2}{c}{$n=50$} & \multicolumn{2}{c}{$n=100$} & \multicolumn{2}{c}{$n=500$} \\
$\tau$ & Cens. & Parameter & Bias & Cov. & Bias & Cov. & Bias & Cov. \\
\hline
2        & 20\% & $\beta_1$  &  0.000  & 95.5  &  0.003  & 95.6  &  0.000  & 94.7 \\
               && $\beta_2$  &  0.015  & 95.2  &  0.002  & 94.7  &  0.000  & 94.7 \\
               && $\gamma_1$ & -0.003  & 95.2  &  0.009  & 94.4  &  0.001  & 95.0 \\
               && $A$        & -0.006  & 94.9  & -0.008  & 94.7  & -0.002  & 94.3 \\
               && $S$        &  0.003  & 96.2  & -0.001  & 95.5  & -0.001  & 95.1 \\
               && $r$        &  0.010  & 95.7  &  0.001  & 95.5  &  0.001  & 94.4 \\
2        & 50\% & $\beta_1$  & -0.010  & 95.7  &  0.000  & 94.6  &  0.001  & 95.4 \\
               && $\beta_2$  &  0.031  & 95.2  &  0.007  & 93.6  & -0.001  & 95.1 \\
               && $\gamma_1$ & -0.018  & 95.1  &  0.012  & 94.4  &  0.007  & 94.8 \\
               && $A$        & -0.003  & 94.4  & -0.005  & 94.4  & -0.002  & 94.8 \\
               && $S$        &  0.000  & 97.6  &  0.003  & 95.4  &  0.000  & 95.0 \\
               && $r$        &  0.022  & 95.4  &  0.007  & 95.1  &  0.000  & 95.1 \\
$\infty$ & 20\% & $\beta_1$  & -0.004  & 95.1  &  0.000  & 94.5  &  0.001  & 94.9 \\
               && $\beta_2$  &  0.013  & 95.6  &  0.002  & 93.8  & -0.002  & 95.0 \\
               && $\gamma_1$ & -0.008  & 95.0  &  0.008  & 94.3  &  0.004  & 94.6 \\
               && $A$        &  0.001  & 94.8  &  0.003  & 95.1  & -0.001  & 95.0 \\
               && $S$        &  0.005  & 96.3  &  0.003  & 94.7  &  0.001  & 95.1 \\
               && $r$        &  0.014  & 95.5  &  0.007  & 95.4  &  0.000  & 95.1 \\
$\infty$ & 50\% & $\beta_1$  & -0.007  & 95.6  & -0.005  & 94.3  & -0.001  & 94.7 \\
               && $\beta_2$  &  0.028  & 95.4  &  0.009  & 93.2  &  0.002  & 94.6 \\
               && $\gamma_1$ & -0.027  & 94.7  &  0.012  & 92.9  &  0.003  & 94.8 \\
               && $A$        & -0.009  & 94.6  & -0.001  & 94.7  & -0.003  & 94.7 \\
               && $S$        & -0.002  & 97.1  &  0.001  & 95.2  &  0.000  & 95.4 \\
               && $r$        &  0.017  & 95.5  &  0.010  & 95.1  &  0.002  & 94.7 \\
\hline
\end{tabular}

{\footnotesize
Cens., censored proportion; Bias, median bias; Cov., empirical coverage percentage for 95\% confidence interval; $A = A(0)$; $S=S(t^{(1)}_{0.5}\,|\, x^{(1)})$; $r = r(x^{(1)},x^{(2)})$.}
\end{center}
\end{table}

\begin{table}[p]
\begin{center}
\caption{Normal (true efficient) bias and coverage\label{tab:reseff}}
\centering
\begin{tabular}{cccrrrrrr}
\hline
&&& \multicolumn{2}{c}{$n=50$} & \multicolumn{2}{c}{$n=100$} & \multicolumn{2}{c}{$n=500$} \\
$\tau$ & Cens. & Parameter & Bias & Cov. & Bias & Cov. & Bias & Cov. \\
\hline
2        & 20\% & $\beta_1$  & -0.005 & 94.2 & -0.001 & 94.4 & -0.001 & 94.7 \\
               && $\beta_2$  &  0.009 & 93.9 &  0.002 & 93.1 &  0.001 & 94.6 \\
               && $\gamma_1$ & -0.011 & 91.8 &  0.004 & 91.8 &  0.002 & 94.9 \\
               && $A$        & -0.009 & 94.1 & -0.002 & 94.4 &  0.000 & 94.3 \\
               && $S$        &  0.000 & 95.3 &  0.000 & 94.3 & -0.001 & 94.5 \\
               && $r$        &  0.012 & 95.2 &  0.004 & 94.8 &  0.001 & 94.6 \\
2        & 50\% & $\beta_1$  & -0.013 & 95.2 & -0.006 & 94.8 &  0.000 & 95.3 \\
               && $\beta_2$  &  0.037 & 93.8 &  0.003 & 93.7 & -0.001 & 95.0 \\
               && $\gamma_1$ & -0.030 & 92.2 &  0.004 & 93.2 & -0.001 & 93.9 \\
               && $A$        & -0.015 & 94.7 & -0.003 & 94.8 & -0.001 & 95.0 \\
               && $S$        & -0.002 & 97.3 &  0.005 & 95.5 &  0.001 & 95.5 \\
               && $r$        &  0.020 & 96.3 &  0.008 & 95.9 &  0.001 & 94.8 \\
$\infty$ & 20\% & $\beta_1$  & -0.006 & 94.1 & -0.001 & 94.3 & -0.001 & 94.5 \\
               && $\beta_2$  &  0.017 & 93.1 & -0.002 & 93.4 &  0.000 & 95.1 \\
               && $\gamma_1$ &  0.000 & 91.5 &  0.003 & 93.2 &  0.005 & 94.2 \\
               && $A$        & -0.008 & 94.4 & -0.005 & 94.2 &  0.000 & 95.2 \\
               && $S$        &  0.000 & 95.2 &  0.004 & 94.4 &  0.000 & 94.6 \\
               && $r$        &  0.010 & 95.2 &  0.005 & 94.6 &  0.001 & 94.9 \\
$\infty$ & 50\% & $\beta_1$  & -0.012 & 95.2 &  0.000 & 94.1 &  0.000 & 95.0 \\
               && $\beta_2$  &  0.027 & 94.2 & -0.001 & 92.6 & -0.002 & 94.6 \\
               && $\gamma_1$ & -0.022 & 92.1 &  0.008 & 92.0 &  0.000 & 94.0 \\
               && $A$        & -0.011 & 94.1 & -0.008 & 94.5 & -0.003 & 94.7 \\
               && $S$        &  0.001 & 96.9 &  0.005 & 94.4 &  0.002 & 95.0 \\
               && $r$        &  0.021 & 95.8 &  0.007 & 95.3 &  0.001 & 95.1 \\
\hline
\end{tabular}

{\footnotesize
Cens., censored proportion; Bias, median bias; Cov., empirical coverage percentage for 95\% confidence interval; $A = A(0)$; $S=S(t^{(1)}_{0.5}\,|\, x^{(1)})$; $r = r(x^{(1)},x^{(2)})$.}
\end{center}
\end{table}

\newpage

\begin{table}[p]
\begin{center}
\caption{Log-rank standard errors \label{tab:lreff}}
\centering
\begin{tabular}{cccrrrrrr}
\hline
&&& \multicolumn{2}{c}{$n=50$} & \multicolumn{2}{c}{$n=100$} & \multicolumn{2}{c}{$n=500$} \\
$\tau$ & Cens. & Parameter & SE & SEE & SE & SEE & SE & SEE \\
\hline
2        & 20\% & $\beta_1$  & 0.148  & 0.142 & 0.099 & 0.099 & 0.045 & 0.044 \\
               && $\beta_2$  & 0.239  & 0.261 & 0.189 & 0.185 & 0.087 & 0.085 \\
               && $\gamma_1$ & 0.237  & 0.256 & 0.174 & 0.179 & 0.081 & 0.082 \\
               && $A$        & 0.252  & 0.223 & 0.168 & 0.156 & 0.071 & 0.071 \\
2        & 50\% & $\beta_1$  & 0.186  & 0.182 & 0.127 & 0.124 & 0.056 & 0.055 \\
               && $\beta_2$  & 0.264  & 0.301 & 0.216 & 0.210 & 0.099 & 0.097 \\
               && $\gamma_1$ & 0.296  & 0.327 & 0.208 & 0.230 & 0.107 & 0.105 \\
               && $A$        & 0.355  & 0.290 & 0.219 & 0.203 & 0.095 & 0.092 \\
$\infty$ & 20\% & $\beta_1$  & 0.149  & 0.142 & 0.100 & 0.099 & 0.044 & 0.044 \\
               && $\beta_2$  & 0.240  & 0.256 & 0.189 & 0.183 & 0.085 & 0.085 \\
               && $\gamma_1$ & 0.236  & 0.252 & 0.175 & 0.177 & 0.082 & 0.081 \\
               && $A$        & 0.245  & 0.221 & 0.165 & 0.156 & 0.072 & 0.071 \\
$\infty$ & 50\% & $\beta_1$  & 0.200  & 0.183 & 0.126 & 0.124 & 0.056 & 0.055 \\
               && $\beta_2$  & 0.268  & 0.298 & 0.215 & 0.208 & 0.099 & 0.097 \\
               && $\gamma_1$ & 0.306  & 0.328 & 0.215 & 0.230 & 0.106 & 0.105 \\
               && $A$        & 0.361  & 0.293 & 0.230 & 0.204 & 0.094 & 0.093 \\
\hline
\end{tabular}

{\footnotesize
SE, standard error of estimates; SEE, median of estimated standard error.}
\end{center}
\end{table}

\begin{table}[p]
\begin{center}
\caption{Gehan standard errors \label{tab:geff}}
\centering
\begin{tabular}{cccrrrrrr}
\hline
&&& \multicolumn{2}{c}{$n=50$} & \multicolumn{2}{c}{$n=100$} & \multicolumn{2}{c}{$n=500$} \\
$\tau$ & Cens. & Parameter & SE & SEE & SE & SEE & SE & SEE \\
\hline
2        & 20\% & $\beta_1$  & 0.137 & 0.142 & 0.096 & 0.099 & 0.044 & 0.044 \\
               && $\beta_2$  & 0.229 & 0.253 & 0.177 & 0.178 & 0.080 & 0.080 \\
               && $\gamma_1$ & 0.243 & 0.249 & 0.178 & 0.174 & 0.077 & 0.077 \\
               && $A$        & 0.244 & 0.223 & 0.166 & 0.156 & 0.072 & 0.070 \\
2        & 50\% & $\beta_1$  & 0.174 & 0.182 & 0.123 & 0.123 & 0.054 & 0.054 \\
               && $\beta_2$  & 0.250 & 0.289 & 0.207 & 0.201 & 0.092 & 0.091 \\
               && $\gamma_1$ & 0.294 & 0.318 & 0.229 & 0.225 & 0.102 & 0.100 \\
               && $A$        & 0.337 & 0.289 & 0.229 & 0.201 & 0.093 & 0.091 \\
$\infty$ & 20\% & $\beta_1$  & 0.139 & 0.144 & 0.098 & 0.099 & 0.044 & 0.044 \\
               && $\beta_2$  & 0.229 & 0.255 & 0.180 & 0.178 & 0.079 & 0.080 \\
               && $\gamma_1$ & 0.242 & 0.247 & 0.176 & 0.174 & 0.078 & 0.077 \\
               && $A$        & 0.245 & 0.223 & 0.161 & 0.156 & 0.071 & 0.070 \\
$\infty$ & 50\% & $\beta_1$  & 0.177 & 0.181 & 0.127 & 0.123 & 0.055 & 0.054 \\
               && $\beta_2$  & 0.252 & 0.289 & 0.207 & 0.202 & 0.092 & 0.092 \\
               && $\gamma_1$ & 0.301 & 0.320 & 0.234 & 0.223 & 0.101 & 0.100 \\
               && $A$        & 0.335 & 0.286 & 0.226 & 0.201 & 0.093 & 0.091 \\
\hline
\end{tabular}

{\footnotesize
SE, standard error of estimates; SEE, median of estimated standard error.}
\end{center}
\end{table}

\begin{table}[p]
\begin{center}
\caption{Normal (true, efficient) standard errors \label{tab:effeff}}
\centering
\begin{tabular}{cccrrrrrr}
\hline
&&& \multicolumn{2}{c}{$n=50$} & \multicolumn{2}{c}{$n=100$} & \multicolumn{2}{c}{$n=500$} \\
$\tau$ & Cens. & Parameter & SE & SEE & SE & SEE & SE & SEE \\
\hline
2        & 20\% & $\beta_1$  & 0.141 & 0.137 & 0.097 & 0.096 & 0.043 & 0.043 \\
               && $\beta_2$  & 0.240 & 0.236 & 0.178 & 0.169 & 0.079 & 0.078 \\
               && $\gamma_1$ & 0.244 & 0.217 & 0.176 & 0.157 & 0.074 & 0.073 \\
               && $A$        & 0.242 & 0.216 & 0.164 & 0.153 & 0.071 & 0.069 \\
2        & 50\% & $\beta_1$  & 0.181 & 0.179 & 0.120 & 0.121 & 0.054 & 0.054 \\
               && $\beta_2$  & 0.285 & 0.274 & 0.201 & 0.193 & 0.089 & 0.088 \\
               && $\gamma_1$ & 0.319 & 0.292 & 0.215 & 0.206 & 0.100 & 0.096 \\
               && $A$        & 0.343 & 0.285 & 0.217 & 0.201 & 0.092 & 0.090 \\
$\infty$ & 20\% & $\beta_1$  & 0.142 & 0.136 & 0.098 & 0.096 & 0.043 & 0.043 \\
               && $\beta_2$  & 0.247 & 0.235 & 0.179 & 0.169 & 0.078 & 0.077 \\
               && $\gamma_1$ & 0.249 & 0.216 & 0.171 & 0.156 & 0.074 & 0.073 \\
               && $A$        & 0.246 & 0.213 & 0.165 & 0.152 & 0.069 & 0.069 \\
$\infty$ & 50\% & $\beta_1$  & 0.177 & 0.179 & 0.124 & 0.121 & 0.053 & 0.054 \\
               && $\beta_2$  & 0.279 & 0.274 & 0.207 & 0.191 & 0.089 & 0.088 \\
               && $\gamma_1$ & 0.305 & 0.291 & 0.224 & 0.206 & 0.098 & 0.096 \\
               && $A$        & 0.336 & 0.289 & 0.223 & 0.201 & 0.093 & 0.091 \\
\hline
\end{tabular}

{\footnotesize
SE, standard error of estimates; SEE, median of estimated standard error.}
\end{center}
\end{table}

\clearpage

\section{Lung cancer analysis}

Below is the table of estimated coefficients and standard errors using the log-rank, Gehan, and normal weight functions for the lung cancer data presented in the main paper; the estimates and standard errors are very similar in all cases.

\begin{table}[!htbp]
\caption{Lung Cancer}
\centering
\begin{tabular}{ll@{\qquad}rc@{\qquad}rc@{\qquad}rc}
\hline
&       & \multicolumn{2}{c}{Log-rank} & \multicolumn{2}{c}{Gehan} & \multicolumn{2}{c}{Normal} \\
& Group & Est. & SE & Est. & SE  & Est. & SE \\
\hline
Scale & Palliative   &  0.000 &   --- &  0.000 &   --- &  0.000 &   --- \\
      & Surgery      & -2.645 & 0.174 & -2.644 & 0.174 & -2.623 & 0.173 \\
      & Chemotherapy & -0.537 & 0.267 & -0.487 & 0.272 & -0.536 & 0.219 \\
      & Radiotherapy & -1.075 & 0.104 & -1.076 & 0.103 & -1.010 & 0.111 \\
      & Chemo\&Radio & -1.868 & 0.118 & -1.867 & 0.120 & -1.866 & 0.113 \\[0.2cm]
Shape & Palliative   &  0.000 &   --- &  0.000 &   --- &  0.000 &   --- \\
      & Surgery      &  0.308 & 0.195 &  0.296 & 0.194 &  0.298 & 0.168 \\
      & Chemotherapy &  0.040 & 0.115 & -0.006 & 0.100 &  0.053 & 0.093 \\
      & Radiotherapy &  0.296 & 0.072 &  0.271 & 0.069 &  0.216 & 0.073 \\
      & Chemo\&Radio &  0.943 & 0.173 &  0.944 & 0.157 &  0.944 & 0.148 \\
\hline
\end{tabular}
\end{table}

\end{document}